\documentclass[a4paper,oneside,11pt]{article}
\usepackage[pdftex]{graphics}
\usepackage{graphicx}
\usepackage{mdframed}
\usepackage{amsfonts}
\usepackage{cancel}
\usepackage{placeins}
\usepackage{sectsty}
\usepackage{amsthm}
\usepackage{amssymb}
\usepackage{appendix}

\date{\vspace{-5ex}}

\newcommand{\captionfonts}{\footnotesize}

\usepackage{amsmath}
\usepackage{mdframed}
\usepackage{subfig}
\usepackage{cancel}
\usepackage{enumitem}
\usepackage{color}
\usepackage{mathtools}
\usepackage{pdfpages}
\usepackage{mathrsfs}
\usepackage{enumitem}

\theoremstyle{plain}

\theoremstyle{definition}
\newtheorem{theorem}{Theorem}[section]
\newtheorem{definition}{Definition}[section]

\newtheorem{lemma}{Lemma}[section]

\theoremstyle{remark}

\makeatletter  
\long\def\@makecaption#1#2{%
  \vskip\abovecaptionskip
  \sbox\@tempboxa{{\captionfonts #1: #2}}%
  \ifdim \wd\@tempboxa >\hsize
    {\captionfonts #1: #2\par}
  \else
    \hbox to\hsize{\hfil\box\@tempboxa\hfil}%
  \fi
  \vskip\belowcaptionskip}
\makeatother   

\begin{document}
\title{Test Models for Statistical Inference: Two-Dimensional Reaction Systems Displaying Limit Cycle Bifurcations and Bistability}
\author{Tomislav Plesa, Tom\'{a}\v{s} Vejchodsk\'{y}, Radek Erban}
  \maketitle
	\date{}
	
\begin{abstract}
Theoretical results regarding two-dimensional 
ordinary-differential equations (ODEs) with second-degree 
polynomial right-hand sides are summarized, with an emphasis on 
limit cycles, limit cycle bifurcations and multistability. 
The results are then used for construction of two reaction 
systems, which are at the deterministic level described by 
two-dimensional third-degree kinetic ODEs. The first system 
displays a homoclinic bifurcation, and a coexistence of 
a stable critical point and a stable limit cycle in the phase 
plane. The second system displays a multiple limit cycle 
bifurcation, and a coexistence of two stable limit cycles. 
The deterministic solutions (obtained by solving the kinetic ODEs) 
and stochastic solutions (noisy time-series 
generating by the Gillespie algorithm, and 
the underlying probability distributions obtained by 
solving the chemical master equation (CME)) of the constructed systems 
are compared, and the observed differences highlighted. 
The constructed systems are proposed as test problems 
for statistical methods, which are designed to detect 
and classify properties of given noisy time-series 
arising from biological applications.
\end{abstract}

\section{Introduction}\label{sec:intro}
Given noisy time-series, it may be of practical importance to 
infer possible biological mechanisms underlying the time-series~\cite{NOI1}. 
Mathematically, such statistical inferences correspond to an inverse problem,
consisting of mapping given noisy time-series to compatible reaction networks. 
One way to formulate the inverse problem is as follows. Firstly, 
one obtains deterministic kinetic ordinary-differential equations (ODEs) 
compatible with the stochastic time-series. And secondly, suitable reaction networks
may then be induced from the obtained kinetic ODEs~\cite{Me,Toth1}.
The inverse problem is generally ill-posed~\cite{Me,Toth1}, as 
 more than one suitable reaction networks may be obtained. 
In order to make a progress in solving the 
inverse problem, it is useful to impose further constraints on the kinetic ODEs.
A particular set of constraints on the kinetic ODEs may be obtained by determining the types of the 
deterministic attractors which are `hidden' in the noisy time-series~\cite{NOI1}. 
This may be a challenging task, especially when cycles (oscillations) are observed in the 
time-series. The observed cycles may be present 
in both the deterministic and stochastic models (also known at the stochastic level as 
\emph{noisy deterministic cycles}), or they may be present 
only in the stochastic model (also known as \emph{quasi-cycles}, or 
noise-induced oscillations). Noisy deterministic cycles may arise directly from 
the autonomous kinetic ODEs, or via the time-periodic terms 
present in the nonautonomous kinetic ODEs. Quasi-cycles may 
arise from the intrinsic or extrinsic noise, and have been 
shown to exist near deterministic stable foci, and stable nodes~\cite{NOI4}. For two-species 
reaction systems, quasi-cycles can be further classified 
into those that are unconditionally noise-dependent (but 
dependent on the reaction rate coefficients), and those 
that are conditionally noise-dependent~\cite{NOI4}. Thus, 
a cycle detected in a noisy time-series may at the 
deterministic level generally correspond to a stable 
limit cycle, a stable focus, or a stable node.

In order to detect and classify cycles in noisy time-series, 
several statistical methods have been suggested~\cite{NOI1,NOI2}. 
In~\cite{NOI1}, analysis of the covariance as 
a function of the time-delay, spectral analysis (the Fourier 
transform of the covariance function), and analysis of the 
shape of the stationary probability mass function, have been 
suggested. Let us note that reaction 
systems of the Lotka-Volterra ($x$-factorable~\cite{Me}) type 
are used as test models in~\cite{NOI1}, and that conditionally 
noise-dependent quasi-cycles, which can arise near a stable node, and which 
can induce oscillations in only a subset of species~\cite{NOI4}, 
have not been discussed. In addition to the aforementioned 
statistical methods developed for analysing noisy time-series, 
methods for (locally) studying the underlying stochastic 
processes near the deterministic attractors/bifurcations 
have also been developed~\cite{NOI4,Radek1,Shuohao,NOI3,NOI5,NOI6,NOI7}.

Statistical and analytical methods for studying cycles 
in stochastic reaction kinetics have often been focused 
on deterministically monostable systems which undergo a local 
bifurcation near a critical (equilibrium) point, known as the supercritical 
Hopf bifurcation. We suspect this is partially due to 
simplicity of the bifurcation, and partially due to the fact 
that it is difficult to find two-species reaction systems, 
which are more amenable to mathematical analysis, undergoing 
more complicated bifurcations and displaying bistability 
involving limit cycles. Nevertheless, 
kinetic ODEs arising from biological applications may exhibit
more complicated bifurcations and multistabilites~\cite{Intro6,Intro7,Intro8}. 
Thus, it is of importance to test the available methods 
on simpler test models that display some of the complexities 
found in the applications.

In this paper, we construct two reaction systems that are 
two-dimensional (i.e. they only include two chemical species)
and induce cubic kinetic equations, first of which undergoes 
a global bifurcation known as a convex supercritical 
homoclinic bifurcation, and which displays bistability 
involving a critical point and a limit cycle (which we 
call mixed bistability). The second system undergoes 
a local bifurcation known as a multiple limit cycle 
bifurcation, and displays bistability involving two 
limit cycles (which we call bicyclicity). Aside from 
finding an application as test models for statistical 
inference and analysis in biology, to our knowledge, the 
constructed systems are also the first examples of 
two-dimensional reaction systems displaying the 
aforementioned types of bifurcations and bistabilities.
Let us note that reaction systems with dimensions higher than two,
displaying the homoclinic bifurcation, as well as 
bistabilities involving two limit cycles, have been 
reported in applications~\cite{Intro6,Intro7,Intro8}.

The reaction network corresponding to the first system 
is given by
\begin{align}
& r_1: \;  & \varnothing  &
\xrightarrow[]{ k_1 } s_1 ,  
\; \; \; \; \; \; \; \; \; \; \; \; 
\; \; \; \; \; \; \; \; \; \; \; \; \; \; \;  
\; \; \; \; \; \; \; \, r_7:  
& 
\varnothing  &\xrightarrow[]{ k_7 } s_2,  \nonumber \\
& r_2: \;  & s_1  &
\xrightarrow[]{ k_2 } 2 s_1 ,  \; \; \; \; \; \; \; \; \; \; \; \; 
\; \; \; \; \; \; \; \; \; \; \; \; \; \; \;  \; \; \; \; \; \, r_8:  
& 
s_2  &\xrightarrow[]{ k_8 } \varnothing,  \nonumber \\
& r_3: \;  & 2 s_1  &\xrightarrow[]{ k_3 } 3 s_1, \; \; \; \; \;  
\; \; \; \; \; \; \; \; \; \; \; \; \;  \; \; \; \; \;  
\; \; \; \; \; \; \; \; \;    
\, r_9: & s_1 + s_2  &\xrightarrow[]{ k_9 } s_1 + 2 s_2, \nonumber \\
& r_4: \;  & s_1 + s_2  &\xrightarrow[]{ k_4 } s_2,  \; \; \; \; \; \; \; 
\; \; \; \; \; \; \; \; \; \; \; \; 
\; \; \; \; \; \; \; \; \;  \; \; \; \; 
\; r_{10}:  
&
2 s_2  &\xrightarrow[]{ k_{10} } 3 s_2, \nonumber \\
& r_5: \;  
& 
2 s_1 + s_2  &\xrightarrow[]{ k_5 } s_1 + s_2, \; \; \; \; 
\; \; \; \; \; \; \; \;  \; \; \; \; \;  \; \; \; \; \; \; \; \; \,  
r_{11}: 
& 3 s_2  &\xrightarrow[]{ k_{11} } 2 s_2, \nonumber \\
& 
r_{6}: \;  & s_1 + 2 s_2  &\xrightarrow[]{ k_6 } 2 s_1 + 2 s_2,  \label{eq:homoclinic1net}
\end{align} 
where the two species $s_1$ and $s_2$ react according to the eleven 
reactions $r_1, r_2, \ldots, r_{11}$ under mass-action kinetics, 
with the reaction rate coefficients denoted $k_1, k_2, \ldots, k_{11}$,
and with $\varnothing$ being the zero-species~\cite{Me}. 
A particular choice of the (dimensionless) reaction rate 
coefficients is given by
\begin{align}
k_1 & = 0.01, \; \; \; 
k_2 = 0.9, \; \; \; 
k_3 = 1.55, \; \; \; 
k_4 = 2.6, \; \; \; 
k_5 = 1.2, \; \; \; 
k_6= 1.5, \nonumber \\
k_7 & = 0.01, \; \; \; 
k_8 = 3.6, \; \; \; 
k_9 = 1, \; \; \; 
k_{10} = 2.4, \; \; \; 
k_{11} = 0.8,
\label{eq:homoclinic1example}
\end{align}
while more general conditions on these parameters are derived
later as equations~(\ref{eq:homoclinic1coefficients}) and~(\ref{eq:homoclinic1parameters}).

The reaction network corresponding to the second system 
includes two species $s_1$ and $s_2$ which are subject
the following fourteen chemical reactions $r_1, r_2, \ldots, r_{14}$:
\begin{align}
& r_1: \;  & \varnothing  &
\xrightarrow[]{ k_1 } s_1 ,  \; \; \; \; \; \; \; \; \; \; \; \; 
\; \; \; \; \; \; \; \; \; \; \; \; \; \; \;  \; \; \; \; \; \; \, r_8:  & 
\varnothing  &\xrightarrow[]{ k_8 } s_2,  \nonumber \\
& r_2: \;  & s_1  &\xrightarrow[]{ k_2 } \varnothing, \; \; \; \; \;  
\; \; \; \; \; \; \; \; \; \; \; \; \;  \; \; \; \; \;  \; \; \; \; \; \; \; \; \; \; \; r_9: &  s_2  &\xrightarrow[]{ k_9 } 2 s_2, \nonumber \\
& r_3: \;  & 2 s_1  &\xrightarrow[]{ k_3 } 3 s_1,  \; \; \; \; \; \; \; 
\; \; \; \; \; \; \; \; \; \; \; \; \; \; \; \; \; \; \; \; \;  \; \; \, r_{10}:  & s_1 + s_2  &\xrightarrow[]{ k_{10} } s_1, \nonumber \\
& r_4: \;  & s_1 + s_2  &\xrightarrow[]{ k_4 } 2 s_1 + s_2, \; \; \; \; 
\; \; \; \; \; \; \; \;  \; \; \; \; \;  \; \; \; \; \; \;   
r_{11}: & 2 s_2  &\xrightarrow[]{ k_{11} } 3 s_2, \nonumber \\
& r_5: \;  & 3 s_1  &\xrightarrow[]{ k_5 } 4 s_1,\; \; \; \; 
\; \; \; \; \; \; \; \;  \; \; \; \; \;  \; \; \; \; \; \; \; \;   
\; \; \; \; \; \; r_{12}: & 2 s_1 + s_2  &\xrightarrow[]{ k_{12} } 2 s_1 + 2 s_2, \nonumber \\
& r_6: \;  & 2 s_1 + s_2  &\xrightarrow[]{ k_6 } s_1 + s_2, \; \; \; \; 
\; \; \; \; \; \; \; \;  \; \; \; \; \;  \; \; \; \; \; \; \; \;   
r_{13}: & s_1 + 2 s_2  &\xrightarrow[]{ k_{13} } s_1 +  s_2, \nonumber \\
& r_7: \;  & s_1 + 2 s_2  &\xrightarrow[]{ k_7 } 2 s_2,\; \; \; \; 
\; \; \; \; \; \; \; \;  \; \; \; \; \;  \; \; \; \; \; \; \; \;   
\; \; \; \; \; \; r_{14}: & 3 s_2  &\xrightarrow[]{ k_{14} } 2 s_2, \label{eq:bicyclicXT2net}
\end{align} 
where $k_1, k_2, \ldots, k_{14}$ are the corresponding reaction 
rate coefficients. A particular choice of the 
(dimensionless) reaction coefficients is given by \footnote{Let us note that the limit cycles corresponding to~(\ref{eq:bicyclicXT2net}) are \emph{highly} sensitive to changes in the parameters~(\ref{eq:bicyclicexample}). Thus, during numerical simulations, parameters~(\ref{eq:bicyclicexample}) should \emph{not} be rounded-off. One can also design bicyclic systems which are less parameter sensitive, see Appendix~\ref{app:constructions3}.}
\begin{align}
k_1 & = 2 \times 10^{-7}, \; \; \; 
k_2 = 19.987880407, \; \; \; 
k_3 = 0.019944378, \nonumber \\ 
k_4 & = 0.02003132232, \; \; \; 
k_5 = 2.9 \times 10^{-8}, \; \; \; 
k_6 = 2.000232 \times 10^{-5}, \nonumber \\
k_7 & = 1.45 \times 10^{-8}, \; \; \; 
k_8 = 2 \times 10^{-7}, \; \; \; 
k_9 = 8.38734, \; \; \; 
k_{10} = 0.038389, \nonumber \\
k_{11} & = 0.0215726,  \; \; \; 
k_{12} = 2 \times 10^{-5}, \; \; \; 
k_{13} = 1.571 \times 10^{-6}, \; \; \; 
k_{14} = 10^{-5}, \label{eq:bicyclicexample}
\end{align}
while the general conditions on these parameters are given later
as equations~(\ref{eq:bicyclicXT2coefficients}) and~(\ref{eq:bicyclicXT2parameters}).

In Figure~\ref{fig:introduction}, we display a representative 
noisy-time series generated using the Gillespie stochastic 
algorithm, in Figure~\ref{fig:introduction}(a) for the one-dimensional cubic 
Schl\"{o}gl system~\cite{Schlogl}, which deterministically 
displays two stable critical points (bistationarity~\cite{Toth1}), 
in Figure~\ref{fig:introduction}(b) for the reaction 
network~(\ref{eq:homoclinic1net}) 
with coefficients~(\ref{eq:homoclinic1example}), which 
deterministically displays a stable critical point and 
a stable limit cycle (mixed bistability), and in 
Figure~\ref{fig:introduction}(c) 
for the reaction network~(\ref{eq:bicyclicXT2net}) 
with coefficients~(\ref{eq:bicyclicexample}), which deterministically 
displays two stable limit cycles (bicyclicity).
Several statistical challenges arise. For example, is it possible 
to infer that the upper attractor in Figure~\ref{fig:introduction}(b) 
is a deterministic critical point, while the lower 
a noisy limit cycle? Is it possible to detect one/both 
noisy limit cycles in Figure~\ref{fig:introduction}(c)? 
The answer to the second question is complicated by the fact
that the two deterministic limit cycles in 
Figure~\ref{fig:introduction}(c) are relatively close to each 
other.

\begin{figure}[tb]
\centerline{
\hskip 0mm
\includegraphics[width=0.5\columnwidth]{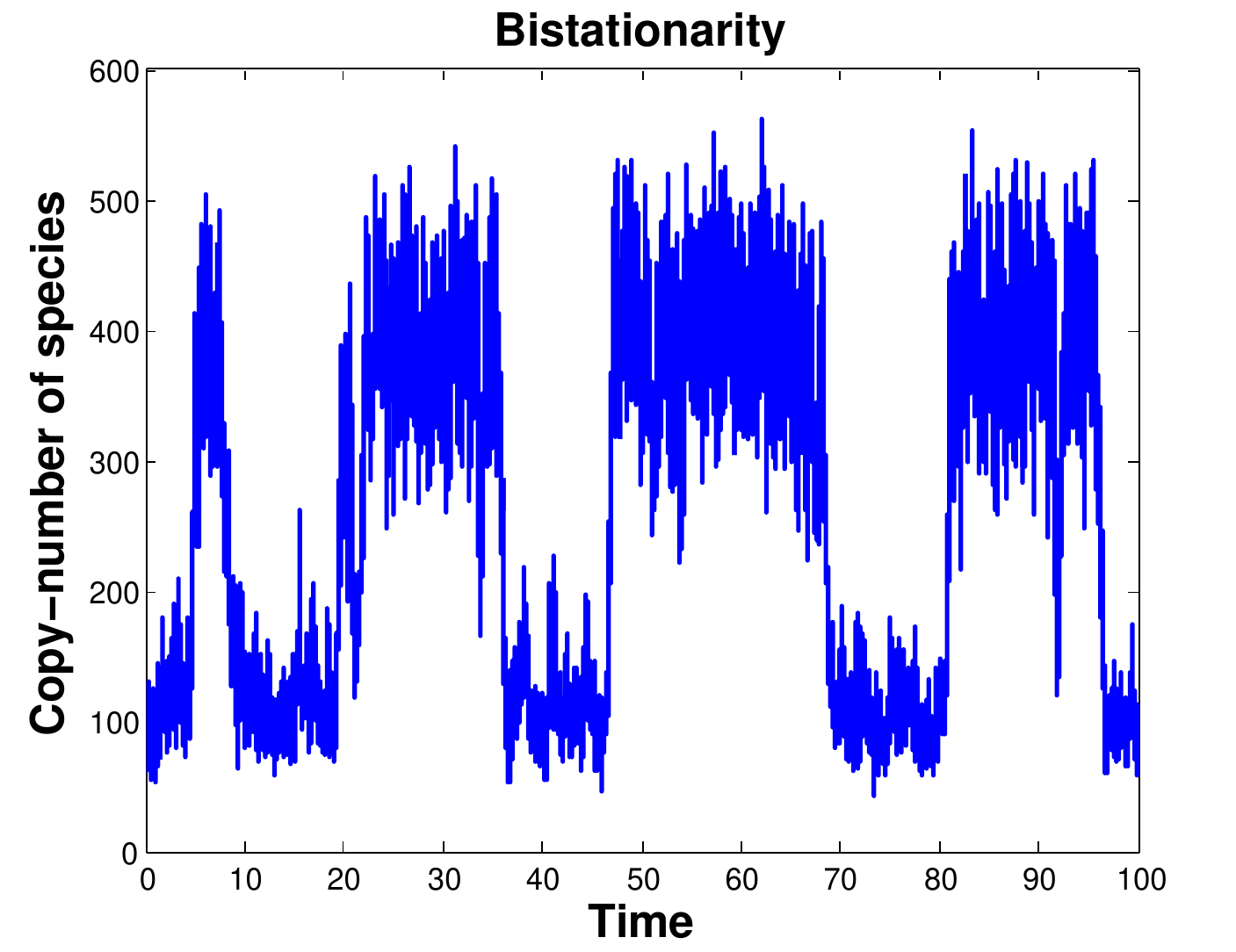}
\hskip 1mm
\includegraphics[width=0.5\columnwidth]{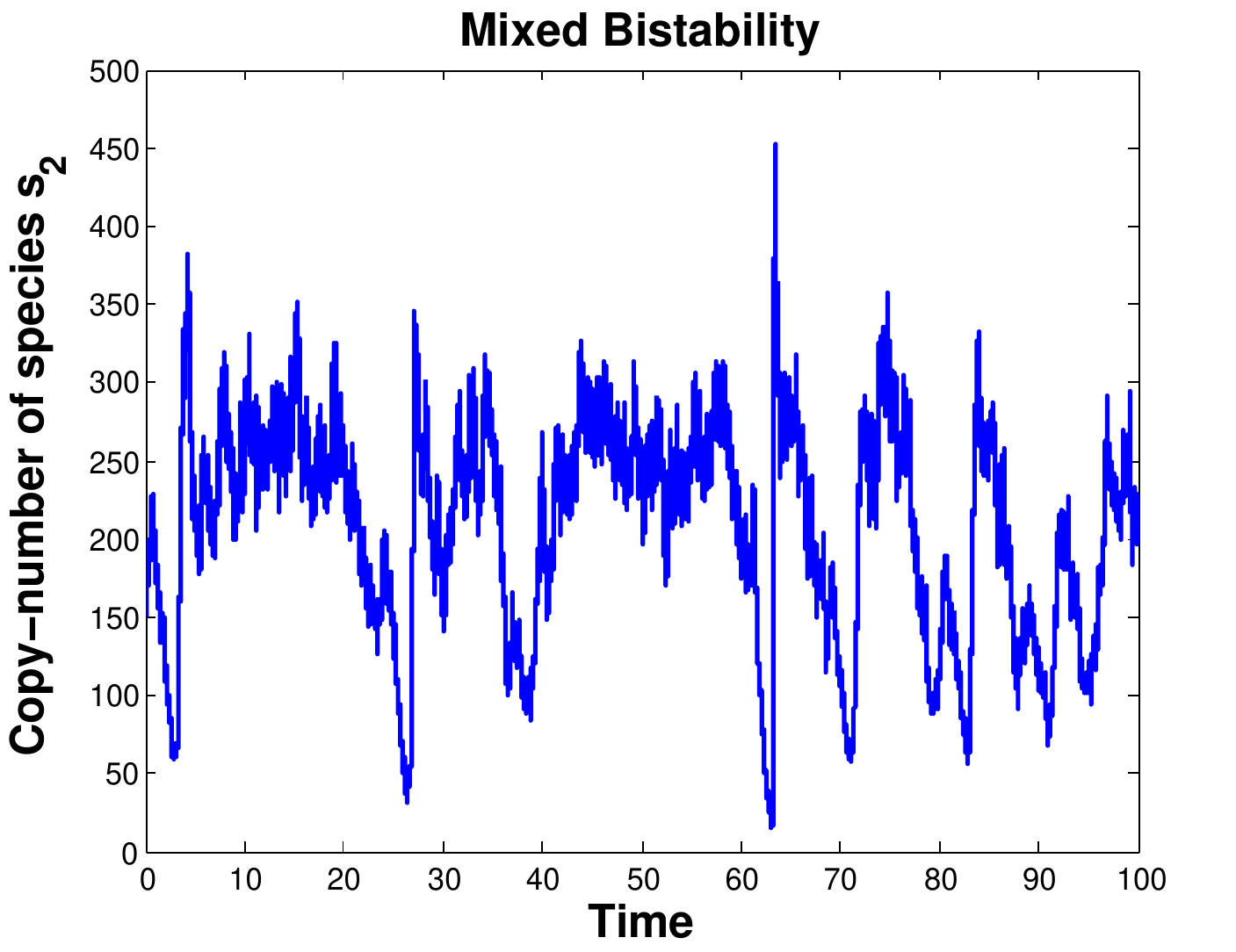}
}
\vskip -5.1cm
\leftline{\hskip -0.3cm (a) \hskip 5.9cm (b)}
\vskip 4.8cm
\centerline{
\hskip -2mm
\includegraphics[width=0.5\columnwidth]{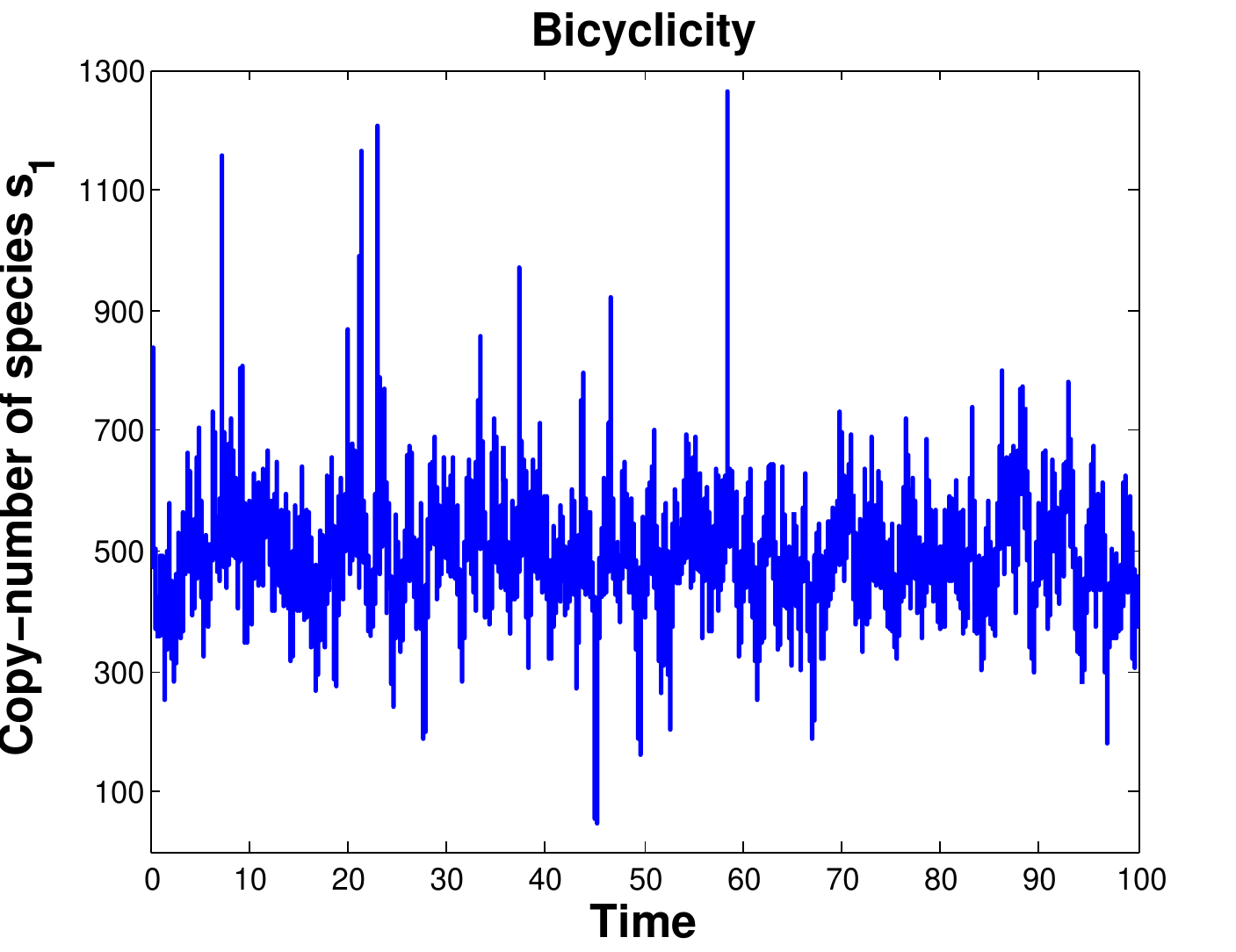}
}
\vskip -5.1cm
\leftline{\hskip 2.8cm (c)}
\vskip 4.3cm
\caption{ 
{\it Panels {\rm (a)}, {\rm (b)} and {\rm (c)} show 
representative sample paths generated using the Gillespie 
stochastic simulation algorithm for the Schl\"{o}gl 
system~{\rm \cite{Schlogl}} with coefficients as in~{\rm \cite{Radek1}}, 
reaction network~$(\ref{eq:homoclinic1net})$ 
with coefficients~$(\ref{eq:homoclinic1example})$ and reactor volume $V = 100$, and 
reaction network~$(\ref{eq:bicyclicXT2net})$ with
coefficients~$(\ref{eq:bicyclicexample})$ and $V = 0.5$, respectively. 
At the deterministic level, the phase planes 
of~$(\ref{eq:homoclinic1net})$ and~$(\ref{eq:bicyclicXT2net})$ 
are shown in Figure~$\ref{fig:phaseplanes}$. The 
deterministic and stochastic time-series,
as well as the probability distributions, are shown in 
Figures~$\ref{fig:homoclinic}$ and~$\ref{fig:nonconcentricbicyclic}$. 
At the deterministic level, a critical point and 
a limit cycle are `hidden' in {\rm (b)}, while two limit 
cycles are `hidden' in {\rm (c)}. 
}
}
\label{fig:introduction}
\end{figure}

The rest of the paper is organized as follows. 
In Section~\ref{sec:properties}, we outline properties 
of the planar quadratic ODE systems, concentrating on 
cycles, cycle bifurcations and multistability. 
There are two reasons for focusing on the planar 
quadratic systems: firstly, the phase plane theory 
for such systems is well-developed~\cite{Gaiko1,Gaiko2}, with 
a variety of concrete examples with interesting 
phase plane configurations~\cite{LC1,Cherkas,Artes}. Secondly, an arbitrary 
planar quadratic ODE system can always be mapped to a kinetic 
one using only an affine transformation - a special property 
not shared with cubic (nor even linear) planar 
systems~\cite{Escher1}. This, together with the available 
nonlinear kinetic transformations which increase the 
polynomial degree of an ODE system by one~\cite{Me}, imply 
that we may map a general planar quadratic system to at most 
cubic planar kinetic system, which may still be biologically 
or chemically relevant. In Section~\ref{sec:constructions}, 
we present the two planar cubic test models which induce 
reaction networks~(\ref{eq:homoclinic1net}) and~(\ref{eq:bicyclicXT2net}), 
and which are constructed starting from suitable planar 
quadratic ODE systems. We also compare the deterministic 
and stochastic solutions of the constructed reaction networks,
and highlight the observed qualitative differences. Finally, 
in Section~\ref{sec:summary}, we provide a summary of the paper.

\section{Properties of two-dimensional second-degree \break polynomial 
ODEs: cycles, cycle bifurcations and multistability}
\label{sec:properties}
Let us consider the two-dimensional second-degree autonomous 
polynomial ODEs 
\begin{align}
\frac{\mathrm{d} x_1}{\mathrm{d} t} 
& = \mathcal{P}_1(x_1,x_2; \, \mathbf{k})  
= k_{1} + k_{2} x_1 + k_{3} x_2 + k_{4}x_1^2 + k_{5} x_1 x_2 + k_{6} x_2^2,
\nonumber \\
\frac{\mathrm{d} x_2}{\mathrm{d} t} 
& = \mathcal{P}_2(x_1,x_2; \, \mathbf{k}) 
= k_{7} + k_{8} x_1 + k_{9} x_2 + k_{10}x_1^2 + k_{11} x_1 x_2 + k_{12} x_2^2,
\label{eq:polynomial}
\end{align}
where 
$\mathcal{P}_i (\,\cdot\,,\,\cdot\,; \, \mathbf{k}): 
\mathbb{R}^2 \to \mathbb{R}$, $i \in \{1,2\},$ 
are the second-degree two-variable polynomial functions, and 
$\mathbf{k} = (k_1, k_2, \dots, k_{12}) \in \mathbb{R}^{12}$ 
is the vector of the corresponding coefficients. 
We assume that $\mathcal{P}_1$ and $\mathcal{P}_2$ are relatively 
prime and at least one is of second-degree. We allow 
coefficients $\mathbf{k}$ to be parameter-dependent,  
$\mathbf{k} = \mathbf{k}(\mathbf{p})$, with 
$\mathbf{p} \in \mathbb{R}^{q}$, $q \ge 0$.

Let us consider two additional properties which system~(\ref{eq:polynomial}) may satisfy:
\begin{enumerate}
\item [(I)]  Coefficients $k_1,$ $k_3,$ $k_6,$ $k_7,$ $k_8,$ $k_{10} \ge 0$,
i.e. $\mathcal{P}_1$ and $\mathcal{P}_2$ are so-called kinetic functions 
(for a rigorous definition see~\cite{Me}). 
\item [(II)] The species concentrations $x_1 = x_1(t)$ and $x_2 = x_2(t)$ are uniformly 
bounded in time for $t \ge 0$ in the nonnegative orthant 
$\mathbb{R}_{\ge}^2$, except possibly for initial conditions 
located on a finite number of one-dimensional subsets of
$\mathbb{R}_{\ge}^{2}$, where infinite-time blow-ups are 
allowed.
\end{enumerate}
The subset of equations~(\ref{eq:polynomial}) satisfying 
properties (I)--(II) are referred to as the 
\emph{deterministic kinetic equations} bounded 
in $\mathbb{R}_{\ge}^2$, and denoted
\begin{align}
\frac{\mathrm{d} x_1}{\mathrm{d} t} 
& = \mathcal{K}_1(x_1,x_2; \, \mathbf{k}(\mathbf{p})), \nonumber \\
\frac{\mathrm{d} x_2}{\mathrm{d} t} 
& = \mathcal{K}_2(x_1,x_2; \, \mathbf{k}(\mathbf{p})). 
\label{eq:kinetic}
\end{align}
In what follows, we discuss only the biologically/chemically
relevant solutions of~(\ref{eq:kinetic}), i.e.
the solutions in the nonnnegative quadrant $\mathbb{R}_{\ge}^2$.
We now summarize some of the definitions and results 
regarding cycles, cycle bifurcations and multistability
(referred to as the so-called exotic phenomena in the 
biological context~\cite{Toth1}) 
for systems~(\ref{eq:polynomial}) and~(\ref{eq:kinetic}). 
Let us note that most of the results have been shown to hold only
for the more general system~(\ref{eq:polynomial}), and may not necessarily hold
for the more restricted system~(\ref{eq:kinetic}).

\emph{Critical points}. A (finite) critical point $(x_1^*(\mathbf{k}),x_2^*(\mathbf{k}))$ 
of system~(\ref{eq:polynomial}) is a solution of the polynomial system 
$\mathcal{P}_1(x_1^*,x_2^*; \, \mathbf{k}) = 0, \mathcal{P}_2(x_1^*,x_2^*; \, \mathbf{k}) = 0$.
Critical points are the time-independent solutions of~(\ref{eq:polynomial}).

\emph{Cycles}. Cycles of~(\ref{eq:polynomial}) are 
closed orbits in the phase plane which are not critical points. 
They can be isolated (limit cycles, and separatrix cycles) or nonisolated 
(a one-parameter continuous family of cycles). Limit cycles 
are the periodic solutions of~(\ref{eq:polynomial}). 
A homoclinic separatrix cycle consists of 
a homoclinic orbit and a critical point of saddle type, with the orbit 
connecting the saddle to itself. On the other hand, a heteroclinic separatrix cycle consists 
of two heteroclinic orbits, and two critical points,
with the orbits connecting the two critical points~\cite{Perko1}. 
Limit cycles of~(\ref{eq:kinetic}) correspond to biological clocks, 
which play an important role in fundamental biological processes, such as the cell cycle, 
the glycolytic cycle and circadian rhythms~\cite{Intro2,Intro3,Intro4}. 

\emph{Cycle bifurcations}. Variations of coefficients $\mathbf{k}$
in~(\ref{eq:polynomial}) may lead to changes in the topology 
of the phase plane (e.g. a change may occur in the number of 
invariant sets or their stability, shape of their region of 
attraction or their relative position). Variation of $\mathbf{k}(\mathbf{p})$ in~(\ref{eq:kinetic}) may be 
interpreted as a variation of the reaction rate 
coefficients $\mathbf{k}$ due to changes in the reactor 
(environment) parameters $\mathbf{p}$, such as the pressure 
or temperature. If the variation causes the system to 
become topologically nonequivalent, such a parameter is 
called a bifurcation parameter, and at the parameter value 
where the topological nonequivalence occurs, a bifurcation is 
said to take place~\cite{Bifur1,Perko1}. Bifurcations in 
the deterministic kinetic equations have been reported in
applications~\cite{Intro2,Intro3,Intro4,Intro1,Intro5,Intro6}.

Bifurcations of limit cycles of~(\ref{eq:polynomial}) 
can be classified into three categories: (i) the Andronov-Hopf 
bifurcation, where a limit cycle is created from a critical point 
of focus or center type, (ii) the separatrix cycle bifurcation, 
where a limit cycle is created from a separatrix cycle, 
and (iii) the multiple limit cycle bifurcation, where a limit 
cycle is created from a limit cycle of multiplicity greater 
than one~\cite{Gaiko1,Perko1}. Let us note that the maximum multiplicity 
of a multiple focus of~(\ref{eq:polynomial}) is 
three, so that at most three local limit cycles can be created 
under appropriate perturbations~\cite{Bautin}. Bifurcations (i) and (iii) are 
examples of local bifurcations, occurring in a neighbourhood 
of a critical point or a limit cycle, while bifurcations (ii) are 
examples of global bifurcations, occuring near a separatrix cycle. 
The following global bifurcations may occur in~(\ref{eq:polynomial}): 
convex homoclinic bifurcations (defined in e.g.~\cite{Homoclinic}),
saddle-saddle (heteroclinic) bifurcations, and the 
saddle-node (heteroclinic) bifurcations on an invariant cycle.
However, concave homoclinic bifurcations, double 
convex, and double concave homoclinic bifurcations, presented in e.g.~\cite{Homoclinic}, 
cannot occur in~(\ref{eq:polynomial}) as a consequence of 
basic properties of planar quadratic ODEs~\cite{LC2,LC3}. 

A necessary condition for the existence of a limit cycle in~(\ref{eq:kinetic}) 
is that $k_4 > 0$ or $k_{12} > 0$~\cite{Me,Toth1}. 
This implies that the induced reaction 
network must contain at least one autocatalytic reaction
of the form $2 s_i \to n s_i + m s_j$, with $n \ge 3$, 
$m \ge 0$, and $i, j \in \{1,2\}$. In the literature,
system~(\ref{eq:kinetic}) has been shown to display the 
following limit cycle bifurcations: Andronov-Hopf bifurcations, 
saddle-node on an invariant cycle, and multiple limit cycle 
bifurcations~\cite{Escher1,Escher2,Escher3}. 
Let us note that some of the reaction systems constructed 
in~\cite{Escher1,Escher2,Escher3}
(e.g. displaying double Andronov-Hopf bifurcation, 
and a saddle-saddle bifucation) are described by 
ODEs of the form~(\ref{eq:kinetic}), but with solutions 
which are generally not bounded in $\mathbb{R}_{\ge}^2$.

\emph{Multistability}. System~(\ref{eq:polynomial}) is 
said to display multistability if the total number of the underlying stable critical points 
and stable limit cycles is greater than one, for a fixed $\mathbf{k}$.
Multistability in~(\ref{eq:kinetic}) corresponds to biological switches, 
which may be classified into reversible or irreversible~\cite{Multistationarity1,Multistationarity2,Intro1}. 
The former switches play an important role in reversible biological 
processes (e.g. metabolic pathways dynamics, and reversible 
differentiation), while the latter in irreversible biological 
processes (e.g. developmental transitions, and apoptosis). 

Multistability can be mathematically classified into 
\emph{pure multistability}, involving attractors 
of only the same type (either only stable critical points, 
or only stable limit cycles), and \emph{mixed multistability}, 
involving at least one stable critical point, and at least 
one stable limit cycle. Pure multistability involving only critical 
points is called \emph{multistationarity}~\cite{Toth1}, while 
we call pure multistability involving only limit cycles 
\emph{multicyclicity}. Mixed bistability, and bicyclicity, 
can be further classified into concentric and nonconcentric. 
Concentric mixed bistability (resp. bicyclicity) occurs when 
the stable limit cycle encloses the stable critical point 
(resp. when the first stable limit cycle encloses the second
stable limit cycle), while nonconcentric when this is not 
the case. Let us note that, for a fixed kinetic ODE system~(\ref{eq:kinetic}), 
multistationarity at some parameter values $\mathbf{k}$, is neither necessary, nor 
sufficient, for cycles at some (possibly other) parameter 
values $\mathbf{k}'$~\cite{Tothosci}.

We now prove that~(\ref{eq:polynomial}) can have at most 
three coexisting stable critical points, i.e.~(\ref{eq:polynomial}) 
can be at most \emph{tristationary}.

\begin{lemma} \label{lemma:tristationarity}
\textit{The maximum number of coexisting stable critical points 
in two-dimensional relatively prime second-degree polynomial 
ODE systems~$(\ref{eq:polynomial})$, with fixed coefficients 
$\mathbf{k}$, is three.}
\end{lemma}

\begin{proof}
Let us assume system~(\ref{eq:polynomial}) has four, 
the maximum number, of real finite critical points. Then, using 
an appropriate centroaffine (linear) transformation~\cite{LC2,LC3}, system~(\ref{eq:polynomial}) can be mapped to 
\begin{align}
\frac{\mathrm{d} x_1}{\mathrm{d} t} 
& = a_1 x_1 (x_1 - 1) + b_1 x_2 (x_2 - 1) + c_1 x_1 x_2, \nonumber \\
\frac{\mathrm{d} x_2}{\mathrm{d} t} 
& = a_2 x_1 (x_1 - 1) + b_2 x_2 (x_2 - 1) + c_2 x_1 x_2, \label{eq:tristationary}
\end{align}
which is topologically equivalent to~(\ref{eq:polynomial}), with 
the critical points located at $A = (0,0)$, $B = (1,0)$, $C = (0,1)$ 
and $D = (\alpha,\beta)$, with $\alpha \ne 0$, $\beta \ne 0$, 
$\alpha + \beta \ne 1$, and the coefficients $c_1, c_2$ given 
by
\begin{align}
c_1 
& = - \frac{\alpha - 1}{\beta} a_1 - \frac{\beta - 1}{\alpha} b_1, 
\nonumber \\
c_2 
& = - \frac{\alpha - 1}{\beta} a_2 - \frac{\beta - 1}{\alpha} b_2. 
\nonumber
\end{align}
The trace and determinant of the Jacobian matrix 
of~(\ref{eq:tristationary}), denoted $\tau$ and $\delta$, 
respectively, evaluated at the four critical points, $A,B,C,D$, 
are given by:
\begin{align}
\tau_A 
& = - (a_1 + b_2), 
& \delta_A \, = \, a_1 b_2 - a_2 b_1, \nonumber \\
\tau_B 
& = a_1 
- a_2 \frac{(\alpha - 1)}{\beta} - b_2\frac{(\alpha + \beta - 1)}{\alpha}, 
& \delta_B \, = \, - \frac{\alpha + \beta - 1}{\alpha} \delta_A, 
\nonumber \\
\tau_C 
& = b_2 - a_1\frac{(\alpha + \beta - 1)}{\beta} - b_1\frac{(\beta - 1)}{\alpha},
& \delta_C \, = \, - \frac{\alpha + \beta - 1}{\beta} \delta_A, 
\nonumber \\
\tau_D 
& = \alpha a_1 + \beta b_2 
- a_2 \frac{\alpha (\alpha - 1)}{\beta} - b_1\frac{\beta (\beta-1)}{\alpha}, 
& \delta_D \, = \, (\alpha + \beta - 1) \delta_A.  
\label{eq:tristationaryjacob} 
\end{align}
System~(\ref{eq:tristationary}) may have three stable 
critical points if and only if the quadrilateral 
$A B C D$, formed by the critical points, is nonconvex, 
and the only saddle critical point is the one located 
at the interior vertex of the quadrilateral~\cite{LC2,LC3}. 
This is the case when $\alpha > 0$, $\beta > 0$, 
$\alpha + \beta < 1$, and $\delta_A > 0$, in which case $A,$ $B,$ 
and $C$ are nonsaddle critical points, while $D$ is a saddle. 
Imposing also the conditions $\tau_A < 0$, $\tau_B < 0$, 
$\tau_C < 0$, ensuring that $A,$ $B,$ and $C$ are stable, 
a solution of the resulting system of algebraic inequalities 
is given by 
$a_1 = 1$, 
$b_1 = -1$, 
$a_2 = 1$, 
$0 < \alpha < 1/2 \left( (1 + 2 \beta) - \sqrt{1 + 8 \beta^2} \right)$, 
$-1 < b_2 < \alpha(- \alpha + \beta + 1)/(\beta (\alpha + \beta -1))$.
\end{proof}
Let us note that if~(\ref{eq:tristationary}) is kinetic, then it 
cannot have three stable critical points. More precisely, 
requiring $b_1 \ge 0$, $a_2 \ge 0$, and $d_A > 0$ and 
$\tau_A < 0$ in~(\ref{eq:tristationaryjacob}), implies $a_1 > 0$ 
and $b_2 > 0$, which further implies $\tau_B >0$, so that $B$ 
is unstable. More generally, the authors have not found a tristationary 
system~(\ref{eq:kinetic}) in the literature (and we conjecture it does not exist).  
On the other hand, bistationary systems~(\ref{eq:kinetic}) do exist
(in fact, even one-dimensional cubic bounded kinetic systems may be 
bistationary, e.g. the Schl\"{o}gl 
model~\cite{Schlogl}, see the time-series shown in 
Figure~\ref{fig:introduction}(a)).
 
The maximum number of stable limit cycles in~(\ref{eq:polynomial}) 
is two, i.e.~(\ref{eq:polynomial}) can be at most \emph{bicyclic}. 
Furthermore, system~(\ref{eq:polynomial}) may also 
display \emph{mixed tristability}, involving one 
stable critical point, and two stable limit cycles. 
This follows from the fact that the maximum number of limit 
cycles in~(\ref{eq:polynomial}) is four, in the unique 
configuration $(3,1)$, a fact only recently proved 
in~\cite{Gaiko2}, solving the second part of Hilbert's 16th 
problem for the quadratic case. If the solutions of~(\ref{eq:polynomial}) are required
to be bounded in the whole $\mathbb{R}^2$, system~(\ref{eq:polynomial}) was conjectured
to have at most two limit cycles~\cite{Perko1,boundedLCs}, and hence have at most one stable 
limit cycle. It remains an open problem if the maximum number of limit cycles 
in the nonnegative orthant of~(\ref{eq:kinetic}) is four or less (we conjecture it is less than four),
and if~(\ref{eq:kinetic}) may be bicyclic. Due to the fact that~(\ref{eq:kinetic}) is (I) 
kinetic (and, hence, nonnegative), and (II) appropriately bounded in $\mathbb{R}_{\ge}^2$, 
additional restrictions are imposed on the boundary
of $\mathbb{R}_{\ge}^2$, and on the critical points
at infinity, complicating the construction of systems~(\ref{eq:kinetic}) 
displaying multistability involving limit cycles. Some results regarding multistability have been obtained
in~\cite{Escher1}: system~(\ref{eq:kinetic}) displaying concentric mixed bistability has been constructed.
The system contains two limit cycles in the nonnegative orthant, and therefore does not exceed 
the conjectured bound on the number of limit cycles in the bounded quadratic systems~\cite{Perko1,boundedLCs}.
While a kinetic system of the form~(\ref{eq:kinetic}) displaying 
concentric bicyclicity has been obtained in~\cite{Escher1}, the system is not bounded in $\mathbb{R}_{\ge}^2$.

\section{Test models: construction and simulations}\label{sec:constructions}
In this section, our aim is to construct two-dimensional kinetic ODEs
bounded in $\mathbb{R}_{\ge}^2$, which display a nonconcentric bistability. 
As highlighted in the previous section,
it may be a difficult task to obtain such systems with at most quadratic terms, 
i.e. in the form~(\ref{eq:kinetic}). To make a progress, in this section,
we allow the two-dimensional kinetic ODEs to contain cubic terms, and we construct two systems.
The first system displays 
a convex homoclinic bifurcation, and mixed bistability, and 
is obtained by modifying a system from~\cite{Me} using the 
results from Appendix~\ref{app:xfactorable}. The second system 
displays a multiple limit cycle bifurcation, and bicyclicity. 
To construct the second system, we use an existing system of 
the form~(\ref{eq:polynomial}), which forms a one-parameter 
family of uniformly rotated vector fields~\cite{LC4,Perko1}, 
and which displays bicyclicity and multiple limit cycle 
bifurcation~\cite{Tung}. We use kinetic transformations 
from~\cite{Me} to map this system, which is of the 
form~(\ref{eq:polynomial}), to a kinetic one, which is of 
the form~(\ref{eq:kinetic}). We then use the results from Appendix~\ref{app:xfactorable}
to map the system of the form~(\ref{eq:kinetic}) to a suitable cubic 
two-dimensional kinetic system.
We also fine-tune the polynomial coefficients in the kinetic ODEs
in such a way that sizes of the two stable limit cycles differ 
by maximally one order of magnitude (a task that can pose 
challenges~\cite{LC1}). As differences may be observed between 
the deterministic and stochastic solutions for parameters 
at which a deterministic bifurcation occurs~\cite{Radek1}, 
we investigate the constructed models for such 
observations. Let us note that an alternative static
(i.e. not dynamic) approach for reaction system construction, 
using only the chemical reaction network theory or kinetic logic, 
provides only conditions for stability of critical points, 
but no information about the phase plane structures~\cite{Feinberg}, 
and is, thus, insufficient for construction of the systems 
presented in this paper. 

\subsection{System 1: homoclinic bifurcation and mixed bistability}
\begin{figure}
\centerline{
\hskip -2mm
\includegraphics[width=0.46\columnwidth]{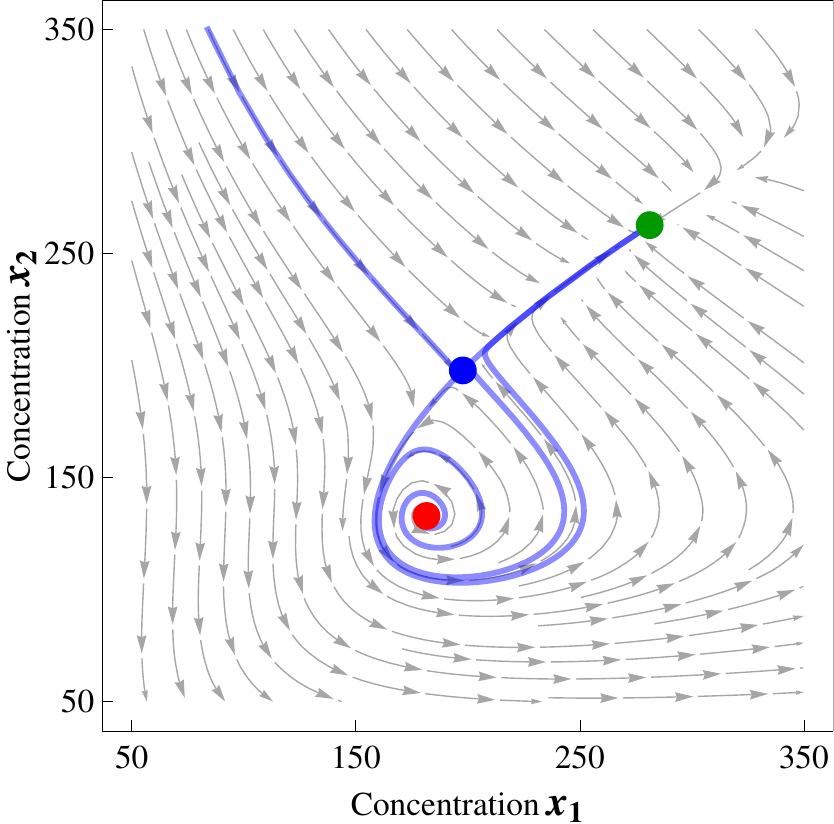}
\hskip 5mm
\includegraphics[width=0.46\columnwidth]{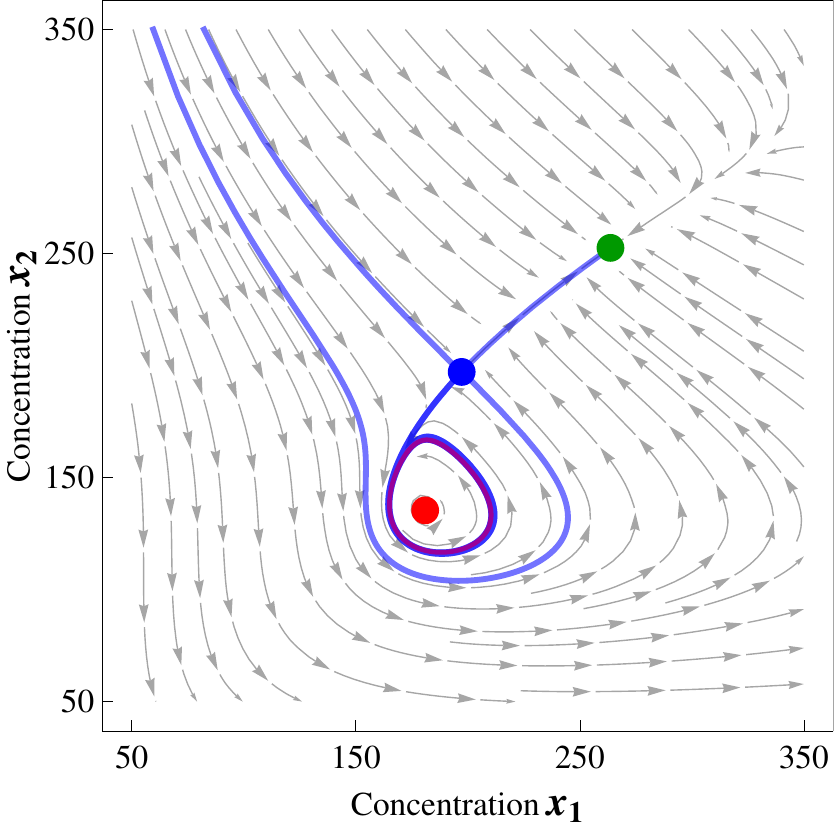}
}
\vskip -6.4cm
\leftline{\hskip 0.2cm (a) {\small$\alpha = 0.05$} \hskip 4.3cm (b) {\small $\alpha = - 0.05$}}
\vskip 6.5cm
\centerline{
\hskip -2mm
\includegraphics[width=0.46\columnwidth]{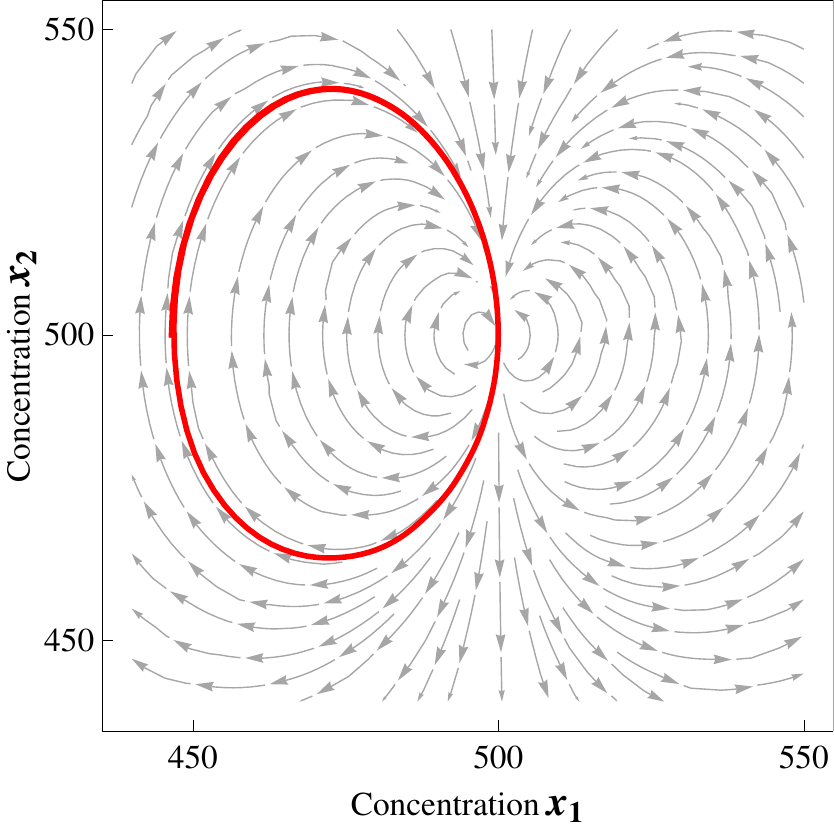}
\hskip 5mm
\includegraphics[width=0.46\columnwidth]{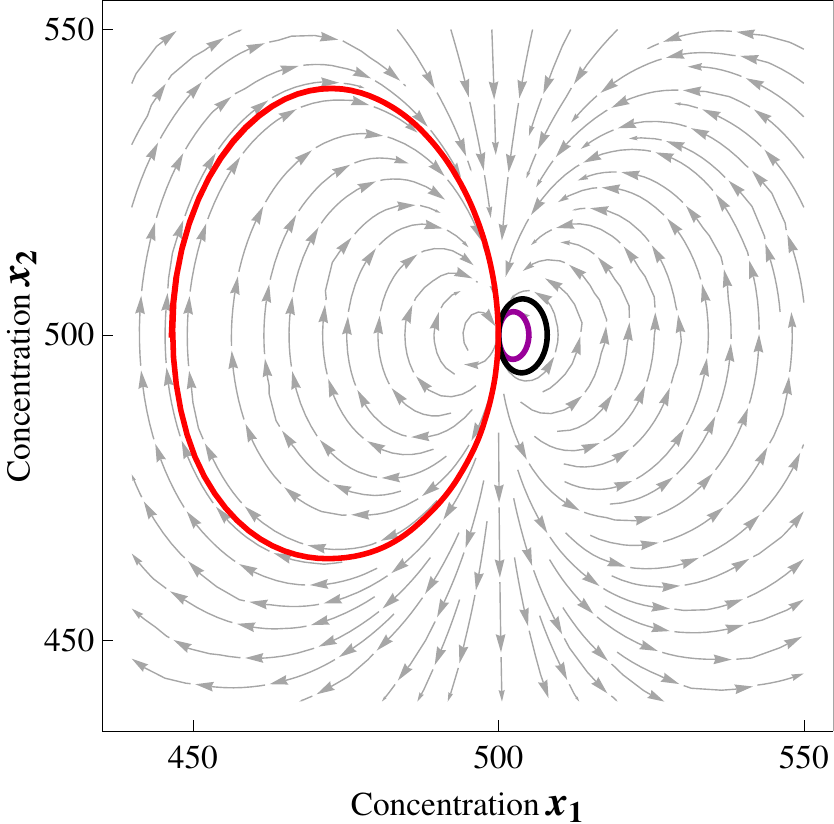}
}
\vskip -6.3cm
\leftline{\hskip 0.2cm (c) {\small $\theta = -0.00147$} \hskip 3.6cm (d) {\small $\theta = -0.00145$}}
\vskip 5.8cm
\caption{ 
{\rm (a)--(b)}
{\it
Phase plane diagrams of system~$(\ref{eq:homoclinic1})$ before and after 
the homoclinic bifurcation. The stable node, saddle, and unstable focus 
are represented as the green, blue and red dots, respectively, the vector 
field as gray arrows, numerically approximated saddle manifolds as blue 
trajectories, and the purple curve in panel {\rm (b)} is the stable limit cycle. 
The parameters appearing in~$(\ref{eq:homoclinic1coefficients})$, and satisfying~$(\ref{eq:homoclinic1parameters})$, are fixed to $a = -0.8$, 
$\mathcal{T}_1 = \mathcal{T}_2 = 2$, $\varepsilon = 0.01$, the reactor volume 
is set to $V = 100$, and the bifurcation parameter $\alpha$ is as shown in 
the panels. }\hfill\break
{\rm (c)--(d)}
{\it Phase plane diagrams of system~$(\ref{eq:bicyclicXT})$ before and after 
the multiple limit cycle bifurcation. The stable limit cycles $L_1$ and 
$L_3$ are shown in red and purple, respectively, while the unstable 
limit cycle $L_2$ is shown in black. The parameters appearing 
in~$(\ref{eq:bicyclicXT2coefficients})$, and 
satisfying~$(\ref{eq:bicyclicXT2parameters})$, are fixed to $a = 1$, 
$b = -1$, $c = 0.5$, $d = 0.08$, $x_1^{*} = -3$, 
$\mathcal{T}_1 = \mathcal{T}_2 = 1000$, $\varepsilon = 0.01$, the reactor 
volume is set to $V = 0.5$, and the bifurcation parameter $\theta$ is as 
shown in the panels.}}
\label{fig:phaseplanes}
\end{figure}

Consider the following deterministic kinetic equations
\begin{align}
\frac{\mathrm{d} x_1}{\mathrm{d} t}  
& = 
k_1 + x_1(k_2 + k_3 x_1 - k_4 x_2 - k_5 x_1 x_2 + k_6 x_2^2),  
\nonumber \\
\frac{\mathrm{d} x_2}{\mathrm{d} t} 
& = 
k_7  + x_2(- k_8 + k_9 x_1 + k_{10} x_2 - k_{11} x_2^2),
\label{eq:homoclinic1}
\end{align}
with the coefficients 
$\mathbf{k} = \mathbf{k}(a, \boldsymbol{\mathcal{T}},\alpha,\varepsilon)$ 
given by
\begin{align}
k_1 & = \varepsilon, \; \; \; \; \; \; \; \; \; \;\; \; \; \; \; \;\; 
\; \; \; \; \;\; \; \; \; \; \;\; \; \; \; \; \;\; \; \; \; \; 
\;\; \; \; \; \; \;\; \; \; \; \; \;\; \; \; \; \; \; \; \; 
k_7  = \varepsilon, \nonumber \\
k_2 & = 
\frac{1}{2} 
\left|\left(3 
\left(\mathcal{T}_2 - \frac{2}{3}\right)(a \mathcal{T}_1 + \mathcal{T}_2) -2 \alpha \mathcal{T}_1 
\right)\right|, \; \; \; 
k_8  = 
|-\mathcal{T}_1 + a \mathcal{T}_2 (\mathcal{T}_2 - 1)|, \nonumber \\
k_3  & = 
\left|-\frac{3}{2} a \left(\mathcal{T}_2 - \frac{2}{3}\right) + \alpha\right|, \; \; \; \; \; \; \; \; \; \; \; \; \; \; \; \; \; \; \; \; \; \; \; \; \; \; \; \; 
k_9 = 1, \nonumber \\
k_4  & =  
\left|1 - \frac{3}{2} (a \mathcal{T}_1 + 2 \mathcal{T}_2)\right|, 
\; \; \; \; \; \; \; \; \; \; \; \; \; \; \; \; \; \; \; \; \; \; \;\; \; \; \; \; \;\;\; \;   
k_{10} = \left|2 a \left(\mathcal{T}_2 - \frac{1}{2} \right)\right|, 
\nonumber \\
k_5 
& = \left|\frac{3}{2} a\right|, \; \; \; \; \; \;\; \; \; \; 
\; \;\; \; \; \; \; \;\; \; \; \; \; \;\; \; \; \; \; \; \;\; \; 
\; \; \; \;\; \; \; \; \; \;\; \; \; \; \; \;\;\;\; \; \; 
k_{11} = |a|, 
\nonumber \\
k_6  & = \frac{3}{2}, 
\label{eq:homoclinic1coefficients}
\end{align}
where $|\cdot|$ denotes the absolute value, and with parameters 
$a$, $\alpha$, $\varepsilon$, $\mathcal{T}_1$, and $\mathcal{T}_2$ 
satisfying
\begin{align}
a & \in (-1,0), \; \; \; 
|\alpha| \ll 1, \; \; \; 1 \ll \varepsilon \le 0, 
\nonumber \\
\mathcal{T}_1 & > \frac{2 \sqrt{3}}{9}, 
\; \; \; 
\mathcal{T}_2 \in \left(\mathrm{max}(1,-a \mathcal{T}_1), 
\frac{2}{3} + \frac{8}{3} a^{-2} (3-a^2) (a + 4 \mathcal{T}_1) \right).
\label{eq:homoclinic1parameters}
\end{align}
The canonical reaction network~\cite{Me} induced by 
system~(\ref{eq:homoclinic1}) is given by~(\ref{eq:homoclinic1net}).

System~(\ref{eq:homoclinic1}) is obtained from 
system~\cite[eq.~(32)]{Me}, which is known to display 
a mixed bistability and a convex supercritical homoclinic 
bifurcation when $\alpha = 0$, $\varepsilon = 0$. We have 
modified \cite[eq.~(32)]{Me} by adding to its right-hand 
side the $\varepsilon$-term from Definition~\ref{def:xft}
(i.e. coefficients $k_1$ and $k_7$ in~(\ref{eq:homoclinic1})), 
thus preventing the long-term dynamics to be trapped on 
the phase plane axes. It can be shown, using 
Theorem~\ref{theorem:Xfact2D}, that choosing a 
sufficiently small $\varepsilon > 0$ 
in~(\ref{eq:homoclinic1coefficients}) does not 
introduce additional positive critical points
in the phase space of~(\ref{eq:homoclinic1}).

In Figures~\ref{fig:phaseplanes}(a) 
and \ref{fig:phaseplanes}(b), we show 
phase plane diagrams of~(\ref{eq:homoclinic1}) before 
and after the bifurcation, respectively, where the 
critical points of the system are shown as the coloured 
dots (the stable node, saddle, and unstable focus are 
shown as the green, blue and red dots, respectively), 
the blue curves are numerically approximated saddle 
manifolds (which at $\alpha = 0$, $\varepsilon = 0$ 
form a homoclinic loop~\cite{Me}), and the purple 
curve in Figure~\ref{fig:phaseplanes}(b) is the 
stable limit cycle that is 
created from the homoclinic separatrix cycle.  Let us note 
that parameter $\alpha$,  appearing 
in~(\ref{eq:homoclinic1coefficients}), controls 
the bifurcation, while parameter $a$ controls 
the saddle-node separation~\cite{Me}.  

In Figures~\ref{fig:homoclinic}(a)--(b) and (d)--(e), 
we show numerical solutions of the initial value problem 
for~(\ref{eq:homoclinic1}) in red, with one initial 
condition in the region of attraction of the node, 
while the other near the unstable focus. The blue 
sample paths are generated by using the Gillespie 
stochastic simulation algorithm on the induced 
reaction network~(\ref{eq:homoclinic1net}), initiated 
near the unstable focus. More precisely, in 
Figures~\ref{fig:homoclinic}(a) and \ref{fig:homoclinic}(d) 
we show the dynamics before the deterministic 
bifurcation, when the node is the globally stable 
critical point for the deterministic model, while 
in Figures~\ref{fig:homoclinic}(b) and 
\ref{fig:homoclinic}(e) we show the dynamics after 
the bifurcation, when the deterministic model 
displays mixed bistability. On the other hand, 
the stochastic model displays relatively frequent 
stochastic switching in Figures~\ref{fig:homoclinic}(a) 
and \ref{fig:homoclinic}(b), when the saddle-node 
separation is relatively small. Let us emphasize 
that the stochastic switching is observed even 
before the deterministic bifurcation. 
In Figures~\ref{fig:homoclinic}(d) and 
\ref{fig:homoclinic}(e), when the saddle-node separation 
is relatively large, the stochastic switching is 
significantly less common, and the stochastic system in the 
state-space is more likely located near the stable 
node. Thus, in Figures~\ref{fig:homoclinic}(d) 
and \ref{fig:homoclinic}(e), the stochastic system 
is less affected by the bifurcation than the 
deterministic system, and, in fact, behaves more 
like the deterministic system before the bifurcation. 
This is also confirmed in Figures~\ref{fig:homoclinic}(c)
and~(f), where we display the $x_2$-marginal stationary probability
mass functions (PMFs) for the smaller and larger saddle-node separations,
respectively, which were obtained by numerically solving the 
chemical master equation (CME)~\cite{VanKampen,Radek2} corresponding to
network~(\ref{eq:homoclinic1net}). Let us note that, by sufficiently increasing 
the saddle-node separation, the left peak in the PMF from
Figure~\ref{fig:homoclinic}(f), corresponding to the deterministic limit cycle,
becomes nearly zero and difficult to detect.

In~\cite{Me2}, we present an algorithm which 
structurally modifies a given reaction network 
under mass-action kinetics, in such a way that 
the deterministic dynamics is preserved, while the 
stochastic dynamics is modified in a controllable state-dependent 
manner. We apply the algorithm on reaction 
network~(\ref{eq:homoclinic1net}), for parameter values similar as in 
Figures~\ref{fig:homoclinic}(d)--(f),
to make the underlying PMF bimodal, so that the underlying sample paths
display stochastic switching between the two deterministic attractors. Furthermore, we also
make the PMF unimodal, and concentrated around the deterministic limit cycle,
so that the underlying sample paths remain near the deterministic limit cycle. 
Meanwhile, we preserve the deterministic dynamics induced by~(\ref{eq:homoclinic1}).

\begin{figure}[!htb]
\centerline{
\hskip -2mm
\includegraphics[width=0.46\columnwidth]{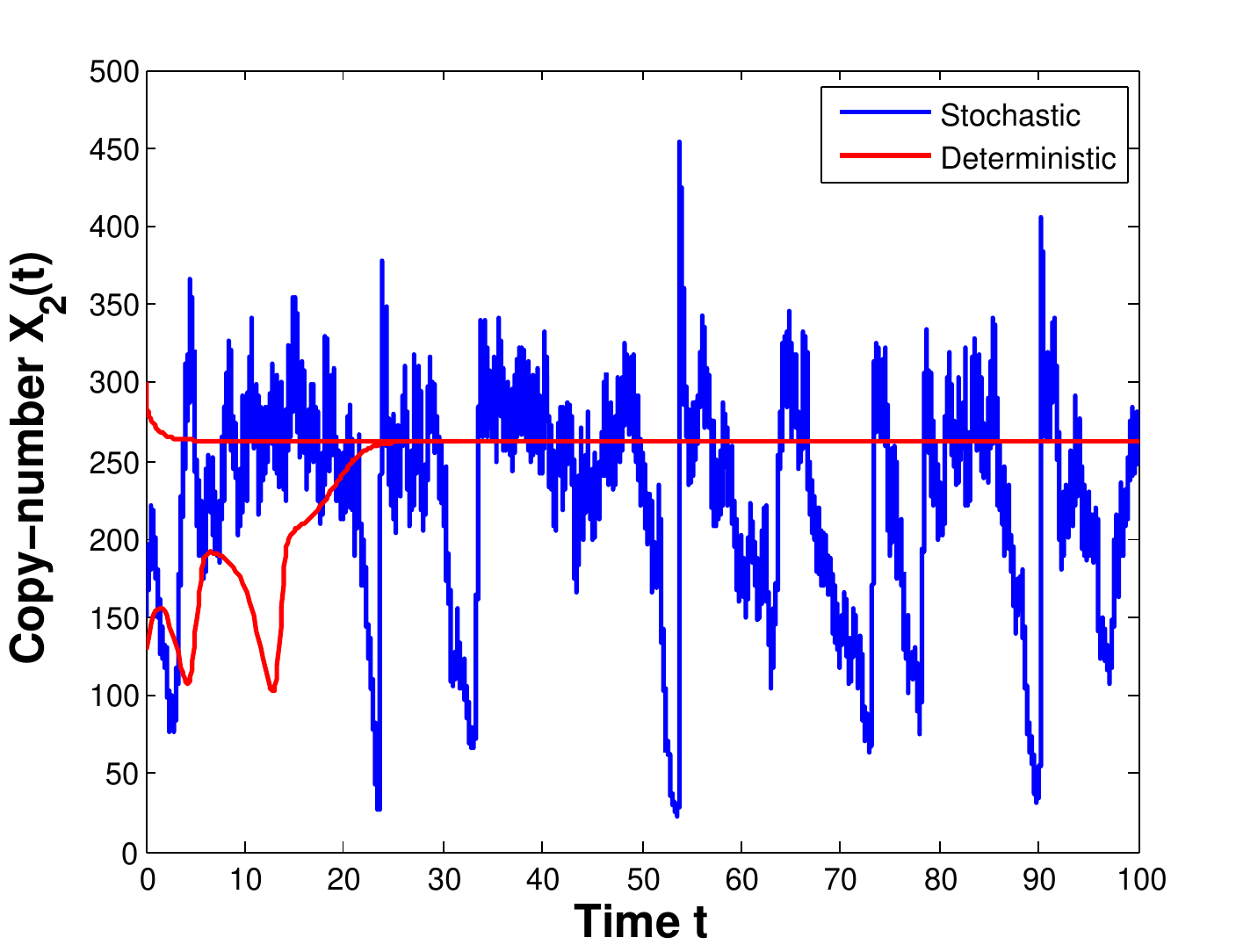}
\hskip 6mm
\includegraphics[width=0.46\columnwidth]{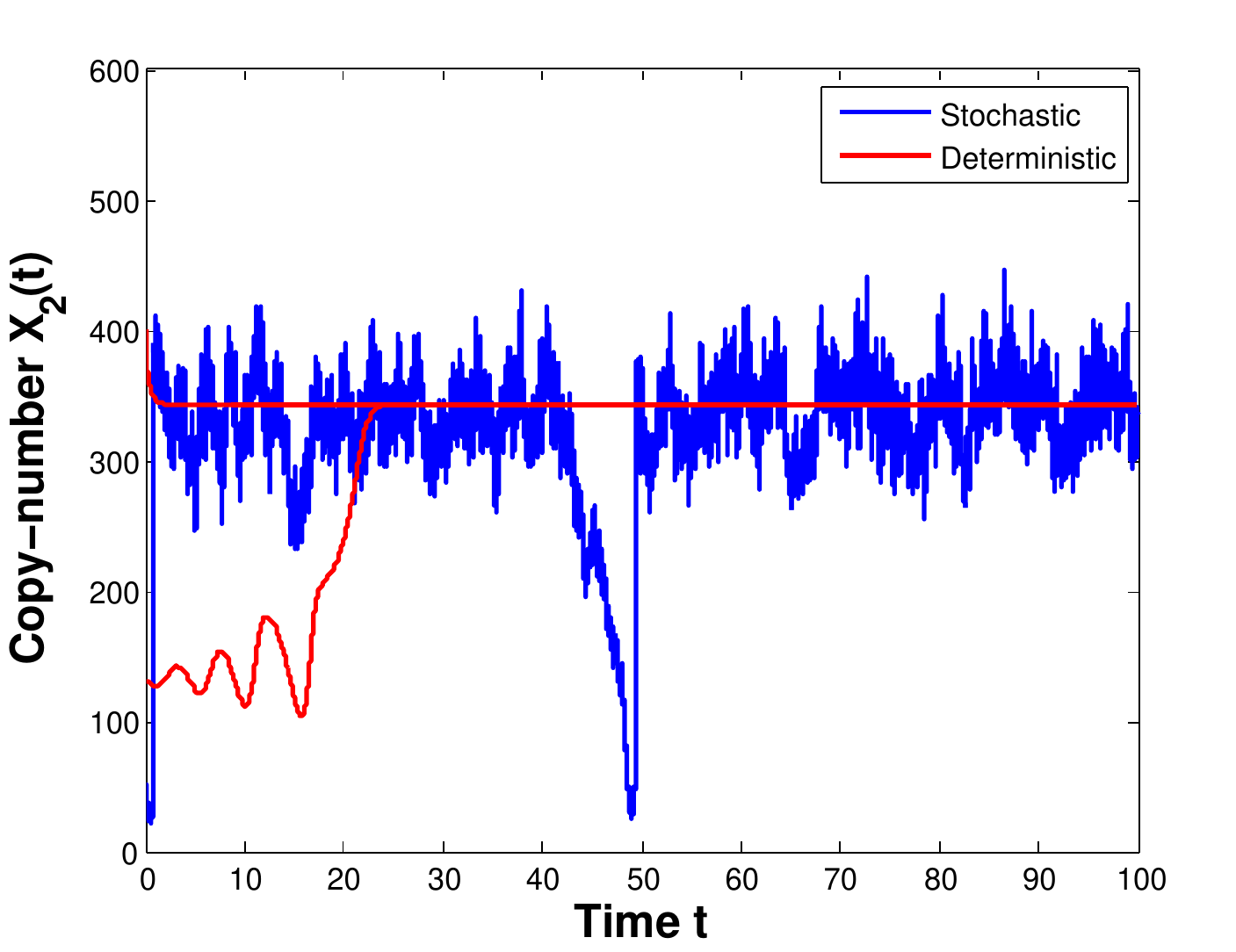}
}
\vskip -5.0cm
\leftline{\hskip 0.4cm (a) {\small $a = -0.8$, $\alpha = 0.05$} 
\hskip 2.9cm (d) {\small $a = -0.65$, $\alpha = 0.05$}}
\vskip 4.8cm
\centerline{
\hskip -2mm
\includegraphics[width=0.46\columnwidth]{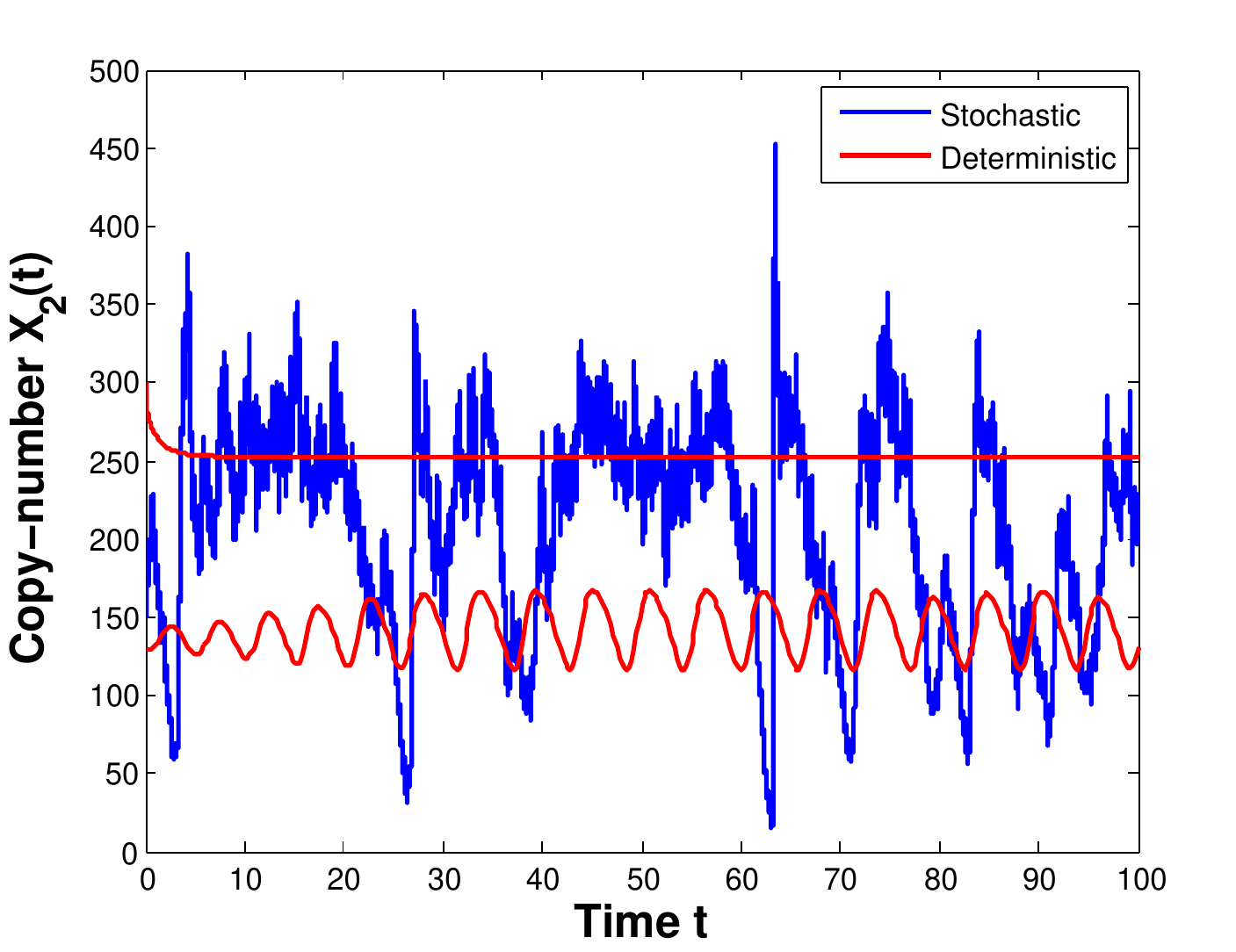}
\hskip 6mm
\includegraphics[width=0.46\columnwidth]{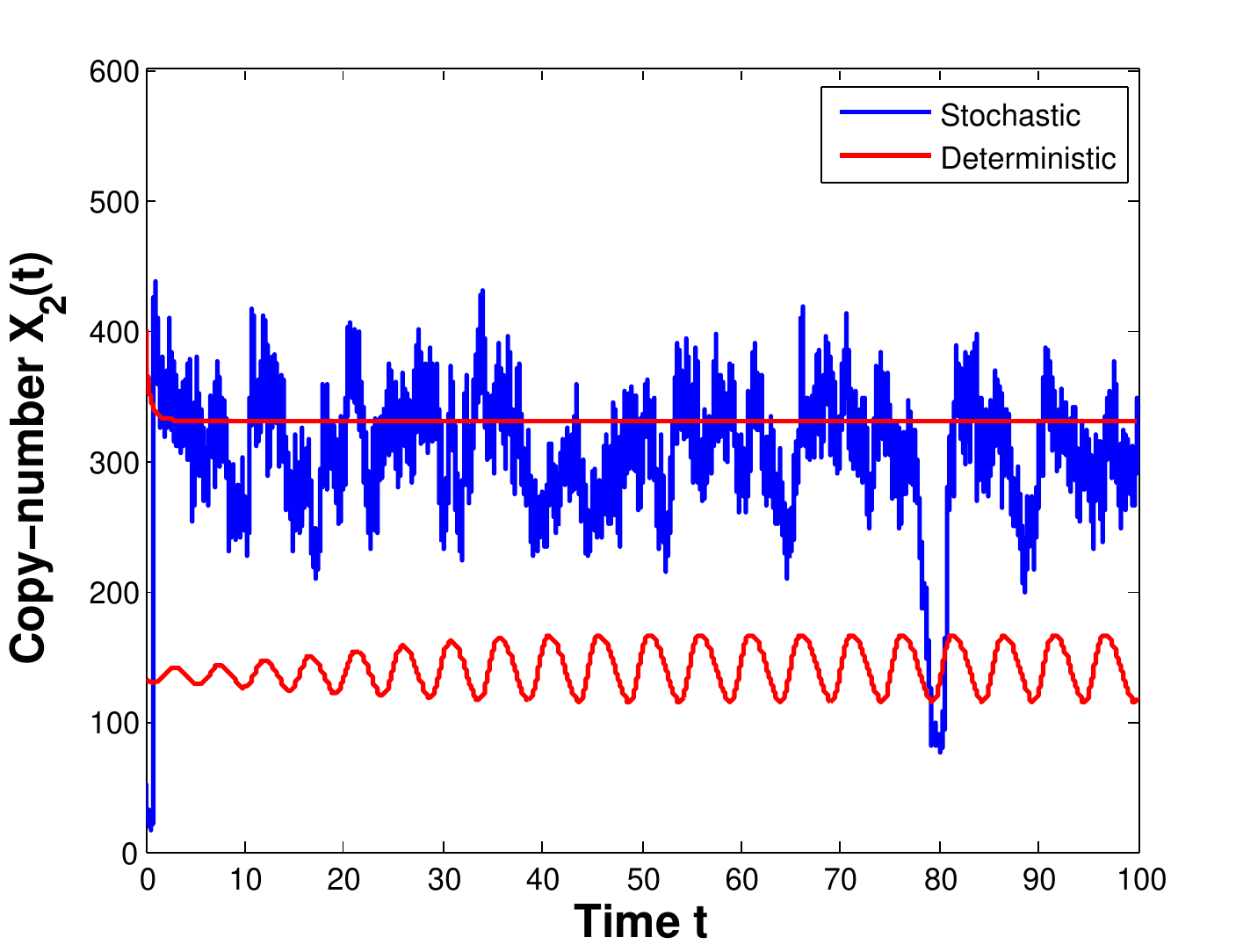}
}
\vskip -5.0cm
\leftline{\hskip 0.4cm (b) {\small $a = -0.8$, $\alpha = -0.05$} 
\hskip 2.6cm (e) {\small $a = -0.65$, $\alpha = -0.05$}}
\vskip 4.8cm
\centerline{
\hskip -2mm
\includegraphics[width=0.46\columnwidth]{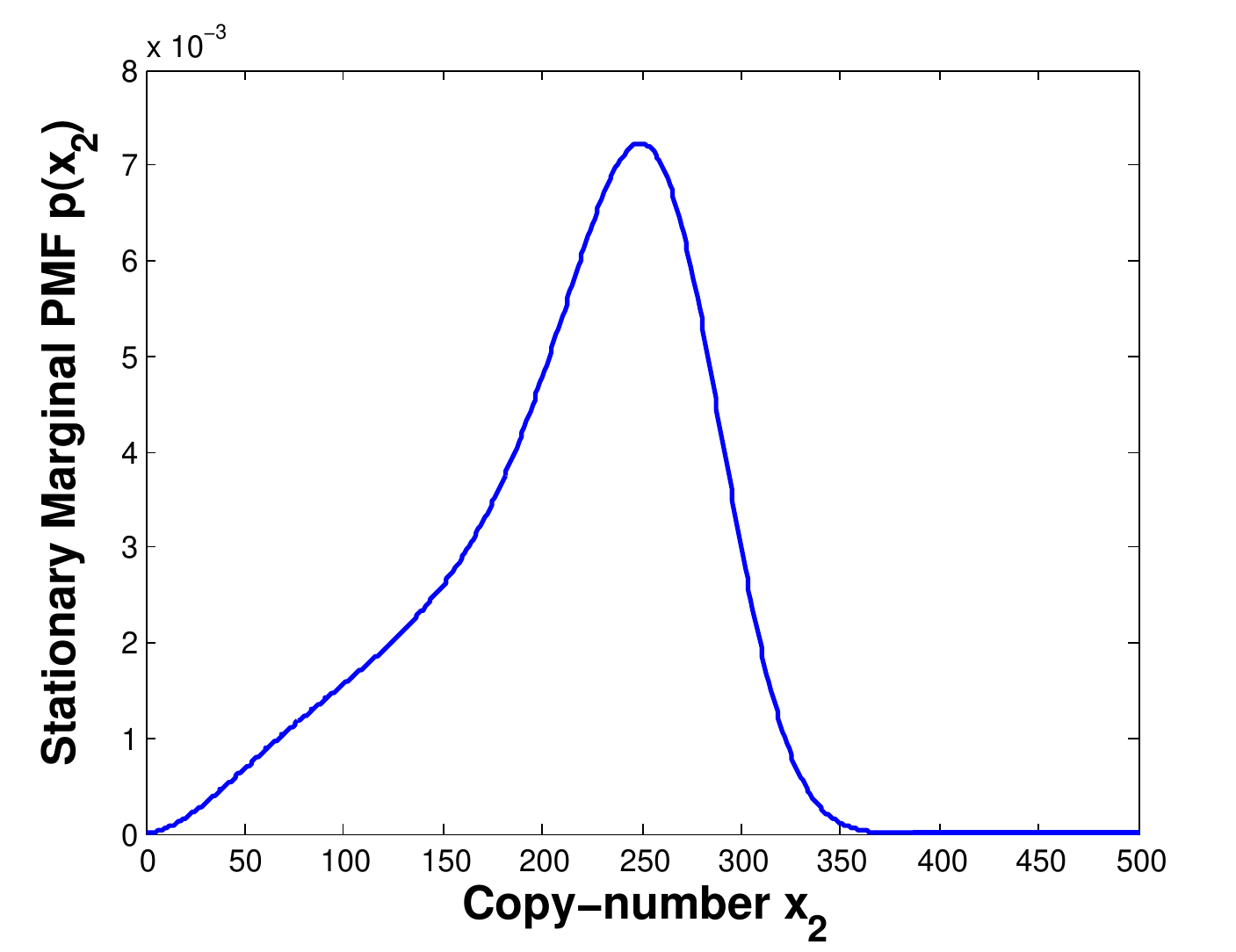}
\hskip 6mm
\includegraphics[width=0.46\columnwidth]{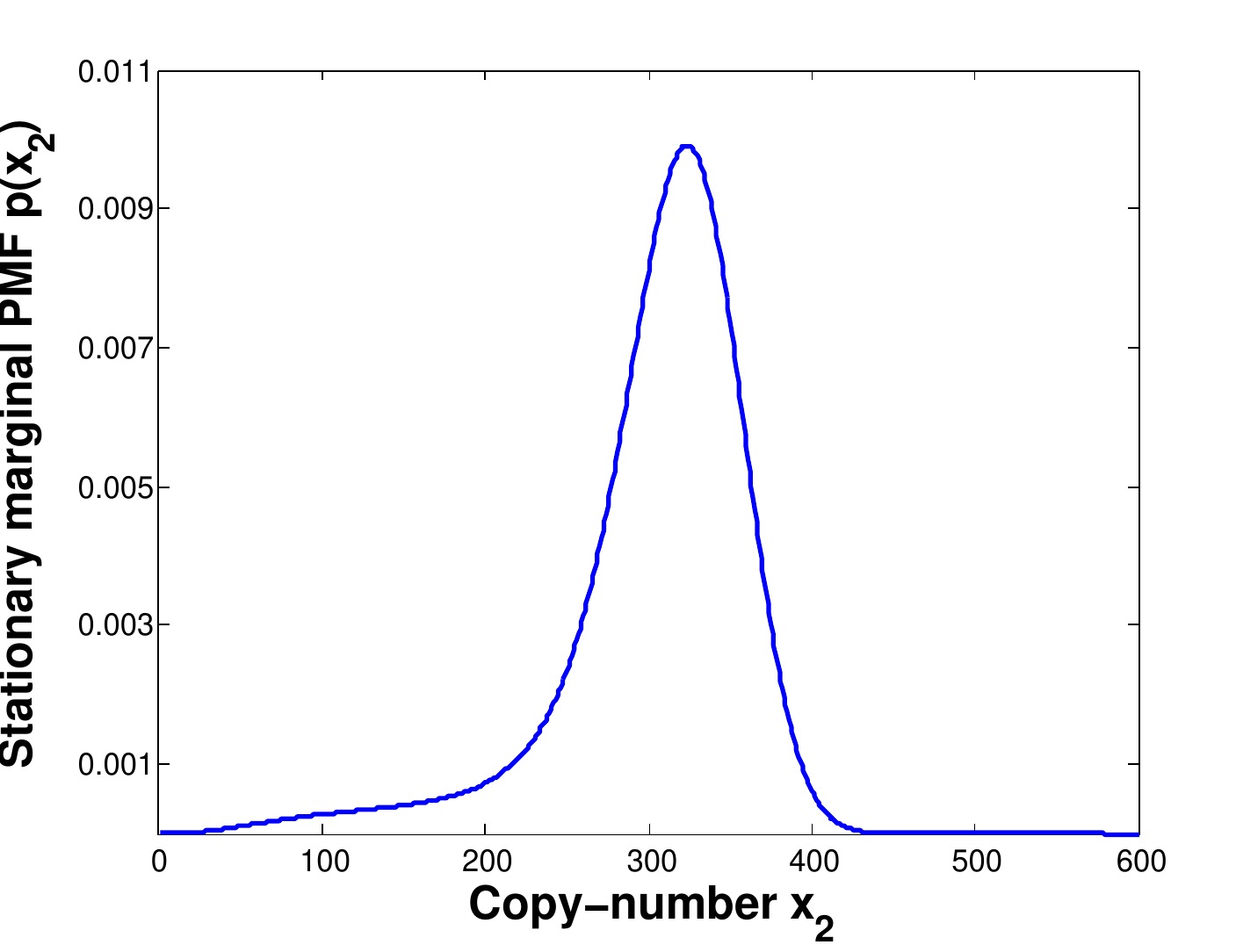}
}
\vskip -5.0cm
\leftline{\hskip 0.4cm (c)
\hskip 6.1cm (f)}
\vskip 4.1cm
\caption{ 
{\it Numerical solutions of system~$(\ref{eq:homoclinic1})$ 
are shown in red. Representative sample paths, 
generated by the Gillespie stochastic simulation algorithm
applied on the corresponding reaction 
network~$(\ref{eq:homoclinic1net})$, are shown in blue.
Probability mass functions (PMFs), obtained
by numerically solving the underlying chemical master equation (CME)
on the bounded domain $(x_1,x_2) \in [0,1000] \times [0,600]$, 
are also shown in blue.\hfill\break}
{\rm (a)--(b)}
{\it The cases before and after the homoclinic bifurcation, 
respectively, for smaller values of $a$, when the limit cycle 
and the stable node are closer together. \hfill\break}
{\rm (d)--(e)}
{\it The cases before and after the homoclinic bifurcation, 
respectively, for larger values of $a$.\hfill\break
{\rm (c) and (f)}
{\it Stationary $x_2$-marginal PMFs.
Parameter values in~(c) and~(f) are
the same as in~(b) and~(e), respectively. \hfill\break}
One of the deterministic solutions is initiated in the 
region of attraction of the node, while the other near 
the focus. The parameters are fixed to 
$\mathcal{T}_1 = \mathcal{T}_2 = 2$, $\varepsilon = 0.01$,  
the reactor volume is set to $V = 100$, with $a$ and $\alpha$ 
as shown in the panels.}}
\label{fig:homoclinic}
\end{figure}

\subsection{System 2: multiple limit cycle bifurcation 
and bicyclicity} \label{sec:constructions2}
Consider the following deterministic kinetic equations
\begin{align}
\frac{\mathrm{d} x_1}{\mathrm{d} t} 
& = k_1 + x_1 (- k_{2} + k_{3} x_1 + k_{4} x_2 + k_{5}x_1^2 - k_{6} x_1 x_2 - k_{7} x_2^2), 
\nonumber \\
\frac{\mathrm{d} x_2}{\mathrm{d} t} 
& = k_8 + x_2 (k_{9} - k_{10} x_1 + k_{11} x_2 + k_{12}x_1^2 - k_{13} x_1 x_2 - k_{14} x_2^2), \label{eq:bicyclicXT}
\end{align}
with coefficients 
$\mathbf{k} = \mathbf{k}(a,b,c,d,x_1^{*},\boldsymbol{\mathcal{T}}, \theta,\varepsilon)$ 
given by
\begin{align}
k_1 & = k_8 = \varepsilon, \nonumber \\
k_{2} & = |- a \mathcal{T}_1 \mathcal{T}_2 \cos(\theta) + [(d (\mathcal{T}_1 + 1) + c \mathcal{T}_2) \mathcal{T}_2  + b (\mathcal{T}_1 + 1) (\mathcal{T}_1 + x_1^{*})] \sin(\theta)|, \nonumber \\
k_{3} & = |a \mathcal{T}_2 \cos(\theta) - [d \mathcal{T}_2 + b (2 \mathcal{T}_1 + x_1^{*} + 1)] \sin(\theta)|, \nonumber \\
k_{4} & = |a \mathcal{T}_1 \cos(\theta) - [d (\mathcal{T}_1 + 1) + 2 c \mathcal{T}_2] \sin(\theta)|, \nonumber \\
k_{5} & = |b \sin(\theta)|, \nonumber \\
k_{6} & = |- a \cos(\theta) + d \sin(\theta)|, \nonumber \\
k_{7} & = |c \sin(\theta)|, \label{eq:bicyclicXT2coefficients}
\end{align}
and if 
$k_{i} = |f(a,b,c,d,x_1^{*},\boldsymbol{\mathcal{T}}) \cos(\theta) 
- g(a,b,c,d,x_1^{*},\boldsymbol{\mathcal{T}}) \sin(\theta)|$, 
then $k_{i + 7} = |f(a,b,c,d,x_1^{*},\boldsymbol{\mathcal{T}}) \sin(\theta) 
+ g(a,b,c,d,x_1^{*},\boldsymbol{\mathcal{T}}) \cos(\theta)|$, 
$i = 2, 3, \ldots, 7$, and with parameters 
$a,$ $b,$ $c,$ $d,$ $x_1^{*},$ $\mathcal{T}_1,$ $\mathcal{T}_2,$ 
$\theta$ and $\varepsilon$ satisfying
\begin{align}
0 \le \varepsilon \ll 1, \, \, \, \, -1 \ll \theta  <  0, \nonumber \\
b < 0, \, \, \, \, d > 0, \, \, \, \, 
a > -\frac{d^2}{4 b}, \, \, \, \, 
0 < c < a + \frac{d^2}{4 b}, \, \, \, \, x_1^{*} < \frac{d^2}{4 b c}, \nonumber \\
a^3 c + b^3 (1 - x_1^{*})^2 \ne 0, \nonumber \\
\mathcal{T}_1 > - x_1^{*},  \, \, \, \, 
 0  < \mathcal{T}_2  <  - \frac{4 a b x_1^{*}}{d^2 (x_1^{*}-1)} (\mathcal{T}_1 + x_1^{*}), 
\nonumber \\
[d (\mathcal{T}_1 + 1) + c \mathcal{T}_2] \mathcal{T}_2 
+ b (\mathcal{T}_1 + 1) (\mathcal{T}_1 + x_1^{*}) < 0. 
\label{eq:bicyclicXT2parameters}
\end{align}
The canonical reaction network induced by 
system~(\ref{eq:bicyclicXT}) is given 
by~(\ref{eq:bicyclicXT2net}). In this section, 
we show that systems~(\ref{eq:bicyclicXT}) 
and~(\ref{eq:bicyclic}) (see below), the latter of 
which is known to display bicyclicity and a multiple 
limit cycle bifurcation, are topologically equivalent 
near the corresponding critical points, provided 
conditions~(\ref{eq:bicyclicXT2parameters}) are satisfied. 

In Figures~\ref{fig:phaseplanes}(c) and \ref{fig:phaseplanes}(d), 
we show the phase plane diagram of~(\ref{eq:bicyclicXT}) 
for a particular choice of the parameters 
satisfying~(\ref{eq:bicyclicXT2parameters}), 
and it can be seen that the system also displays 
bicyclicity and a multiple limit cycle bifurcation, 
with Figures~\ref{fig:phaseplanes}(c) and 
\ref{fig:phaseplanes}(d) showing the cases before 
and after the bifurcation, respectively. In 
Figure~\ref{fig:phaseplanes}(c), the only stable invariant 
set is the limit cycle shown in red, while in 
Figure~\ref{fig:phaseplanes}(d) there are two 
additional limit cycles - a stable one, shown 
in purple, and an unstable one, shown in black. 
The purple, black and red limit cycles are denoted 
in the rest of the paper by $L_1$, $L_2$ and $L_3$,
respectively. At the bifurcation point, $L_1$ and 
$L_2$ intersect.

In order to construct~(\ref{eq:bicyclicXT}), 
let us consider the planar quadratic ODE 
system~\cite{Tung,Escher1} given by
\begin{align}
\frac{\mathrm{d} x_1}{\mathrm{d} t} 
& = \mathcal{Q}_1(x_1,x_2) \cos(\theta)  
- \mathcal{Q}_2(x_1,x_2) \sin(\theta) , \nonumber \\
\frac{\mathrm{d} x_2}{\mathrm{d} t} 
& =  \mathcal{Q}_1(x_1,x_2) \sin(\theta)   
+  \mathcal{Q}_2(x_1,x_2) \cos(\theta), \label{eq:bicyclic}
\end{align}
where
\begin{align}
\mathcal{Q}_1(x_1,x_2) 
& = - a x_1 x_2, \nonumber \\
\mathcal{Q}_2(x_1,x_2) 
& = - b x_1^{*} + b (x_1^{*} + 1) x_1 + d x_2 - b x_1^2 - d x_1 x_2 - c x_2^2, 
\label{eq:bicyclicfunc}
\end{align}
with
\begin{align}
x_1^{*} & < 0, \, \, 
d^2 - 4 b c x_1^{*} < 0, \, \, 
d^2 - 4 b (c - a) < 0, \nonumber \\
\theta d \left(a - b (1 - x_1^{*}) \right) 
& < 0, \, \, \,\theta b d > 0, \, \, \, 
a^3 c + b^3 (1 - x_1^{*})^2 \ne 0. 
\label{eq:bicyclicconditions}
\end{align}

\begin{lemma} \label{lemma:bicyclicity}
\textit{Consider 
system~$(\ref{eq:bicyclic})$--$(\ref{eq:bicyclicconditions})$, 
with the real parameter $\theta \in (-\pi,\pi]$. Function 
$\boldsymbol{\mathcal{P}}(x_1,x_2; \, \theta) 
= ( \mathcal{Q}_1 \cos(\theta)  
- \mathcal{Q}_2 \sin(\theta), \mathcal{Q}_1 \sin(\theta)   
+  \mathcal{Q}_2 \cos(\theta))$ forms a \emph{one-parameter 
family of uniformly rotated vector fields} with the 
rotation parameter $\theta$, and the following results hold:
\begin{enumerate}
\item \emph{Finite critical points.} 
System~$(\ref{eq:bicyclic})$ has two critical 
points in the finite part of the phase plane, 
located at $(1,0)$ and $(x_1^{*},0)$, both of which 
are unstable foci when $|\theta| \ll 1$.
\item \emph{Number and distribution of limit cycles.} 
System~$(\ref{eq:bicyclic})$ has three limit cycles 
in the configuration $(2,1)$ when $|\theta| \ll 1$. The 
focus located at $(1,0)$ is surrounded by two positively 
oriented limit cycles $L_1$ and $L_2$, with the 
unstable limit cycle $L_2$ enclosing the stable limit 
cycle $L_1$, while the focus at $(x_1^{*},0)$ by a single 
negatively oriented stable limit cycle $L_3$.
\item \emph{Dependence of the limit cycles on the rotation 
parameter $\theta$.} There exists a critical value 
$\theta = \theta^{*} < 0$, at which the limit cycles $L_1$ 
and $L_2$ intersect in a semistable, positively oriented 
limit cycle that is stable from the inside, and unstable from 
the outside. As $\theta$ is monotonically increased in 
$(\theta^{*},0)$, the limit cycles $L_2$ and $L_3$ 
monotonically expand, while $L_1$ monotonically contracts.
\end{enumerate}
}
\end{lemma}

\begin{proof}
The statement of the lemma follows from~\cite{Tung,Escher1}, 
and the theory of one-parameter family of uniformly 
rotated vector fields~\cite{LC4,Perko1}.
\end{proof}

In order to map the stable limit cycles of 
system~(\ref{eq:bicyclic}) into the first quadrant, 
and then map the resulting system to a kinetic one, 
having no boundary critical points, let us apply a 
translation transformation $\Psi_{\mathcal{T}}$~\cite{Me}, 
$\boldsymbol{\mathcal{T}} = (\mathcal{T}_1,\mathcal{T}_2) \in \mathbb{R}^2$, 
followed by a perturbed $x$-factorable transformation, 
as defined in Definition~\ref{def:xft}, on 
system~(\ref{eq:bicyclic}), which results in 
system~(\ref{eq:bicyclicXT}) with the 
coefficients~(\ref{eq:bicyclicXT2coefficients}). 

\begin{theorem}\label{lemma:bicyclicityXT}
\textit{Consider the ODE systems~$(\ref{eq:bicyclicXT})$ 
and~$(\ref{eq:bicyclic})$, and assume  
conditions~$(\ref{eq:bicyclicXT2parameters})$ 
are satisfied. Then~$(\ref{eq:bicyclicXT})$ 
and~$(\ref{eq:bicyclic})$ are locally topologically 
equivalent in the neighborhood of the corresponding 
critical points. Furthermore, for sufficiently 
small $\varepsilon > 0$, system~$(\ref{eq:bicyclicXT})$ 
has exactly one additional critical point in 
$\mathbb{R}_{>}^2$, which is a saddle located in 
the neigbhourhood of $(\mathcal{T}_1,0)$.}
\end{theorem}

\begin{proof}
Consider the critical point $(1,0)$ of 
system~(\ref{eq:bicyclic}), which corresponds to 
the critical point $(\mathcal{T}_1+1,\mathcal{T}_2)$ 
of system~(\ref{eq:bicyclicXT}) when $\varepsilon = 0$. 
The Jacobian matrices of~(\ref{eq:bicyclic}),
and~(\ref{eq:bicyclicXT}) with $\varepsilon = 0$, 
evaluated at $(1,0)$, and $(\mathcal{T}_1+1,\mathcal{T}_2)$, 
are respectively given by
\begin{align}
J & = \left(\begin{array}{cc}
-b (x_1^{*} - 1) \sin(\theta) & \; \; -a \cos(\theta)
\\
b (x_1^{*} - 1) \cos(\theta)  & \; \; -a \sin(\theta)
\end{array}
\right), \nonumber \\
J_{\mathcal{X},\mathcal{T}} & = \left(\begin{array}{cc}
-b (x_1^{*} - 1)(\mathcal{T}_1+1) \sin(\theta) 
& \; \; -a (\mathcal{T}_1+1)\cos(\theta)
\\
b  (x_1^{*} - 1) \mathcal{T}_2 \cos(\theta)  
& \; \; -a \mathcal{T}_2 \sin(\theta)
\end{array}
\right). \nonumber
\end{align}
Condition (ii) of \cite[Theorem 3.3]{Me} is satisfied, 
so that the stability of the critical point is preserved 
under the $x$-factorable transformation, but condition 
(iii) is not satisfied. In order for 
$(\mathcal{T}_1+1,\mathcal{T}_2)$ to remain 
focus under the $x$-factorable transformation, 
the discriminant of $J_{\mathcal{X},\mathcal{T}}$ 
must be negative:
\begin{align}
(a \mathcal{T}_2 + b (\mathcal{T}_1 + 1) (x_1^{*} - 1) )^2 (\sin(\theta))^2 - 4 a b (x_1^{*} - 1) (\mathcal{T}_1 + 1) \mathcal{T}_2 & < 0. \label{eq:tempLHS}
\end{align} 
Let us set $\theta = 0$ in~(\ref{eq:tempLHS}), leading to
\begin{align}
- 4 a b (x_1^{*} - 1) (\mathcal{T}_1 + 1) \mathcal{T}_2 & < 0. 
\label{eq:tempLHS2}
\end{align} 
Conditions~(\ref{eq:tempLHS}) and~(\ref{eq:tempLHS2}) are 
equivalent when $|\theta| \ll 1$, since the the sign of 
the function on the LHS of~(\ref{eq:tempLHS}) is 
a continuous function of $\theta$. From 
conditions~(\ref{eq:bicyclicXT2parameters}) it 
follows that $a b < 0$, $x_1^{*} < 0$, and 
$\mathcal{T}_1, \mathcal{T}_2 > 0$, so 
that~(\ref{eq:tempLHS2}) is satisfied. Similar 
arguments show that the second critical point 
of~(\ref{eq:bicyclic}), located at $(x_1^{*},0)$, 
is mapped to an unstable focus of~(\ref{eq:bicyclicXT}), 
if $d > 0$, and if $\mathcal{T}_2$ is bounded as 
given in~(\ref{eq:bicyclicXT2parameters}).

Consider~(\ref{eq:bicyclicXT}) with $\varepsilon = 0$. 
The boundary critical points are located at $(0,0)$, 
$(\mathcal{T}_1,0)$, and $(0,x_{2,\pm}^{*})$, with
\begin{align}
x_{2,\pm}^{*} 
& = \frac{1}{2 c} \left(d (\mathcal{T}_1 + 1) 
+ 2 c \mathcal{T}_2 \pm \sqrt{(\mathcal{T}_1 + 1) 
(d^2 (\mathcal{T}_1 + 1) - 4 b c (\mathcal{T}_1 + x_1^{*}))} \right). \nonumber
\end{align}
Conditions~(\ref{eq:bicyclicXT2parameters}) 
imply that the critical point $(0,0)$ satisfies 
$\mathcal{P}_1(0,0) = - a \mathcal{T}_1 \mathcal{T}_2 < 0$, 
and 
\begin{align}
\mathcal{P}_2(0,0) & = - [d (1 + \mathcal{T}_1) + c \mathcal{T}_2] \mathcal{T}_2 - b (1 + \mathcal{T}_1) (\mathcal{T}_1 + x_1^{*}) >  0, \nonumber
\end{align}
when $\theta = 0$. When $|\theta| \ll 1$, it 
then follows from condition (iv) of \cite[Theorem 3.3]{Me} 
that the critical point is a saddle, and from 
Theorem~\ref{theorem:Xfact2D}, 
condition~(\ref{eq:boundarycondition2}), that it 
is mapped outside of $\mathbb{R}_{\ge}^2$ when 
$\varepsilon \ne 0$. Similar arguments show that,
assuming conditions~(\ref{eq:bicyclicXT2parameters}) 
are true, $(\mathcal{T}_1,0)$ is a saddle that is 
mapped to $\mathbb{R}_{>}^2$ when $\varepsilon \ne 0$, 
and that critical points $(0,x_{2,\pm}^{*})$ are 
real, $x_{2,-}^{*} < 0$, and that $(0,x_{2,+}^{*})$ 
is a saddle that is mapped outside $\mathbb{R}_{\ge}^2$ 
when $\varepsilon \ne 0$. 

Finally, if conditions~(\ref{eq:bicyclicXT2parameters}) 
are satisfied, so are conditions~(\ref{eq:bicyclicconditions}).
\end{proof}

We now consider the kinetic ODEs~(\ref{eq:bicyclicXT}) and the induced reaction network~(\ref{fig:nonconcentricbicyclic}) for a particular set of coefficients~(\ref{eq:bicyclicXT2coefficients}).
We also rescale 
the time according to $t \to 2 \times 10^{-5} \,t$, i.e. we multiply all 
the coefficients $k_1, \ldots k_{14}$ appearing in~(\ref{eq:bicyclicXT}) 
by $ 2 \times 10^{-5}$. On this time-scale, we capture 
dynamical effects relevant for this paper.
In Figures~\ref{fig:nonconcentricbicyclic}(a) 
and \ref{fig:nonconcentricbicyclic}(b) we show numerically 
approximated solutions of the initial value problem 
for~(\ref{eq:bicyclicXT}) before and after the bifurcation, 
respectively. In Figure~\ref{fig:nonconcentricbicyclic}(a), 
the solution is initiated near the unstable focus outside the 
limit cycle $L_3$, and it can be seen that the solution spends 
some time near the unstable focus, followed by an excursion 
that leads it to the stable limit cycle $L_3$, where is then 
stays forever. In Figure~\ref{fig:nonconcentricbicyclic}(b), 
the solutions tend to the limit cycle $L_1$ or $L_3$, depending 
on the initial condition. Let us note that the critical value 
at which the limit cycles $L_1$ and $L_2$ intersect, at the 
deterministic level, is numerically found to be 
$\theta^{*} \approx -0.00146$. 

In Figures~\ref{fig:nonconcentricbicyclic}(c) and 
\ref{fig:nonconcentricbicyclic}(d) we show representative 
sample paths generated by applying the Gillespie stochastic 
simulation algorithm on the reaction 
network~(\ref{eq:bicyclicXT2net}), before and after the 
bifurcation, respectively. One can notice that the stochastic 
dynamics does not appear to be significantly influenced 
by the bifurcation, as opposed to the deterministic 
dynamics. In Figures~\ref{fig:nonconcentricbicyclic}(c) and 
\ref{fig:nonconcentricbicyclic}(d),   
one can notice pulses similar as in
Figure~\ref{fig:nonconcentricbicyclic}(a), that are 
now induced by the intrinsic noise present in the system.

The stationary PMF corresponding to 
network~(\ref{eq:bicyclicXT2net}), for parameter
values as in Figures~\ref{fig:nonconcentricbicyclic}(c) and 
\ref{fig:nonconcentricbicyclic}(d),
accumulates at the boundary of the state-space 
(see also the Keizer paradox~\cite{Keizer}).
While the results from Appendix~\ref{app:xfactorable}
may be used to prevent a PMF from accumulating 
at the boundary, one may need a sufficiently large
reactor volume. For example, for network~(\ref{eq:homoclinic1net}), 
the propensity function~\cite{Radek2} of reactions $r_1$ and $r_7$, for parameter values
taken in this paper (i.e. $\varepsilon = 0.01$ in~(\ref{eq:homoclinic1coefficients}), 
and $V = 100$), takes the value $\varepsilon V = 1$. This is sufficient for 
the underlying PMF to approximately vanish at the boundary of the state-space, 
as demonstrated in Figures~\ref{fig:homoclinic}(c) and~(f). 
On the other hand, for network~(\ref{eq:bicyclicXT2net}), we take $\varepsilon = 0.01$
in~(\ref{eq:bicyclicXT2coefficients}), and $V = 0.5$, so that the propensity function of
$r_1$ and $r_8$ takes the value of only $0.005$. As a consequence, the underlying PMF
accumulates at the boundary of the state-space. Instead of increasing the reactor volume
to prevent this, we instead 
focus on the so-called quasi-stationary PMF under the condition that the species copy-numbers are positive,
$p_{>}(x,y) \equiv p(x,y|x>0,y>0)$. The quasi-stationary PMF
describes well the stochastic dynamics of network~(\ref{eq:bicyclicXT2net}) on the time-scale
of interest, presented in Figures~\ref{fig:nonconcentricbicyclic}(c) and 
\ref{fig:nonconcentricbicyclic}(d). In Figure~\ref{fig:nonconcentricbicyclic}(e), we display
an approximate $x_1$-marginal quasi-stationary PMF $p_{>}(x_1)$, for the same parameter values as in 
Figure~\ref{fig:nonconcentricbicyclic}(d). The quasi-stationary PMF $p_{>}(x_1)$
was obtained by numerically solving the stationary CME corresponding to network~(\ref{eq:bicyclicXT2net}),
on a truncated domain which excludes the boundary of the state-space.

\begin{figure}[t]
\centerline{
\hskip -2mm
\includegraphics[width=0.46\columnwidth]{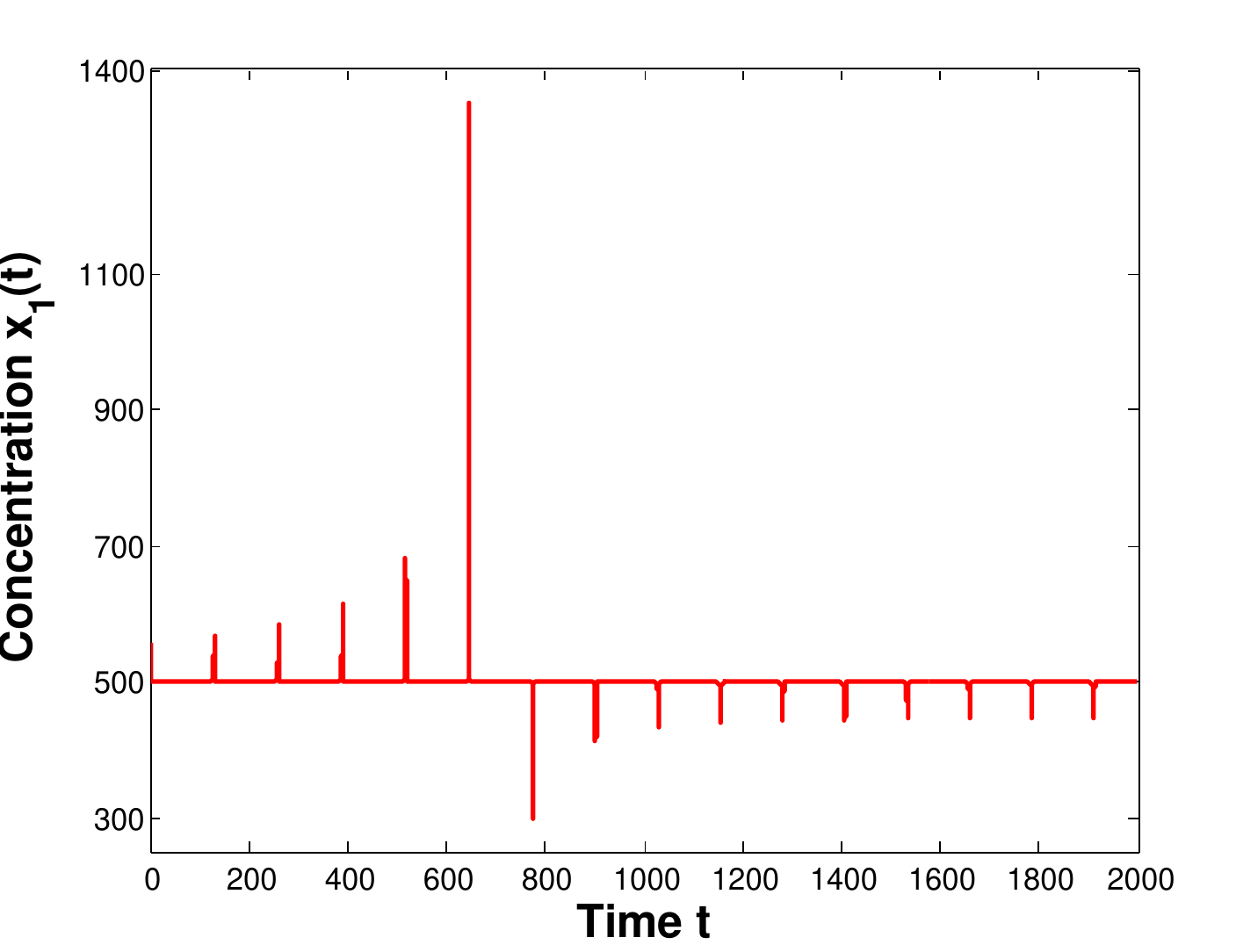}
\hskip 6mm
\includegraphics[width=0.46\columnwidth]{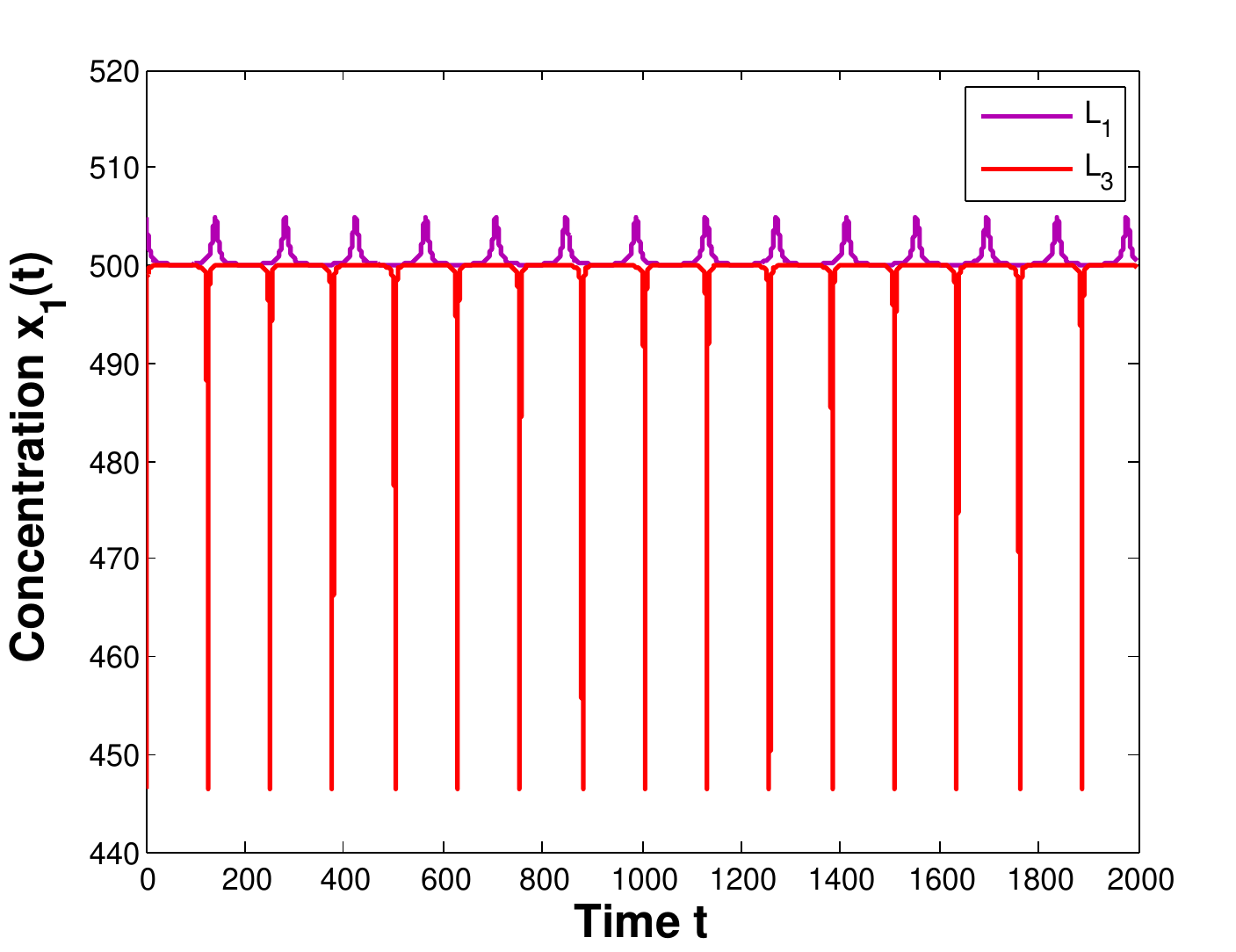}
}
\vskip -5.0cm
\leftline{\hskip 0.4cm (a) {\small $\theta = -0.00147$} 
\hskip 3.8cm (b) {\small $\theta = -0.00145$}}
\vskip 4.8cm
\centerline{
\hskip -2mm
\includegraphics[width=0.46\columnwidth]{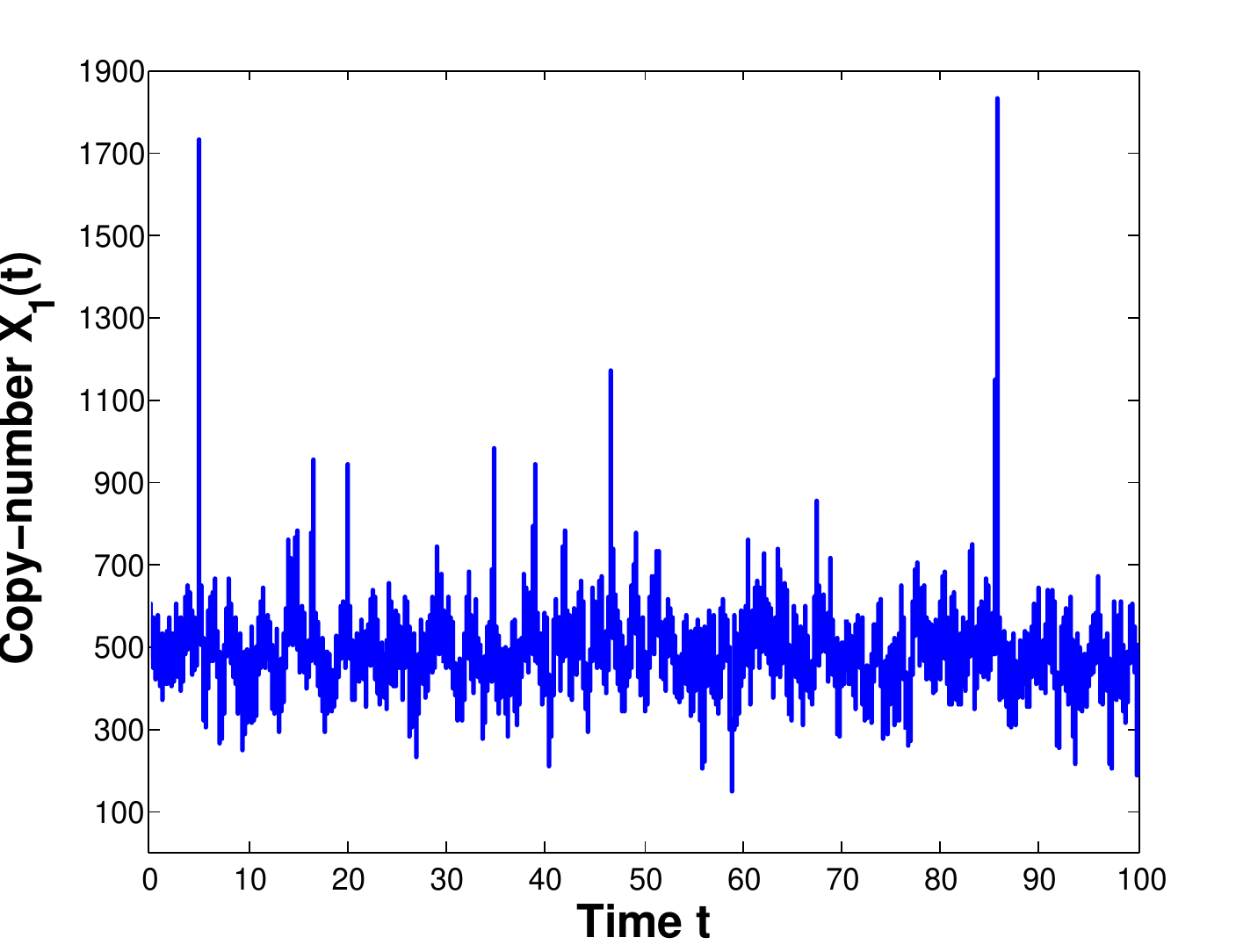}
\hskip 6mm
\includegraphics[width=0.46\columnwidth]{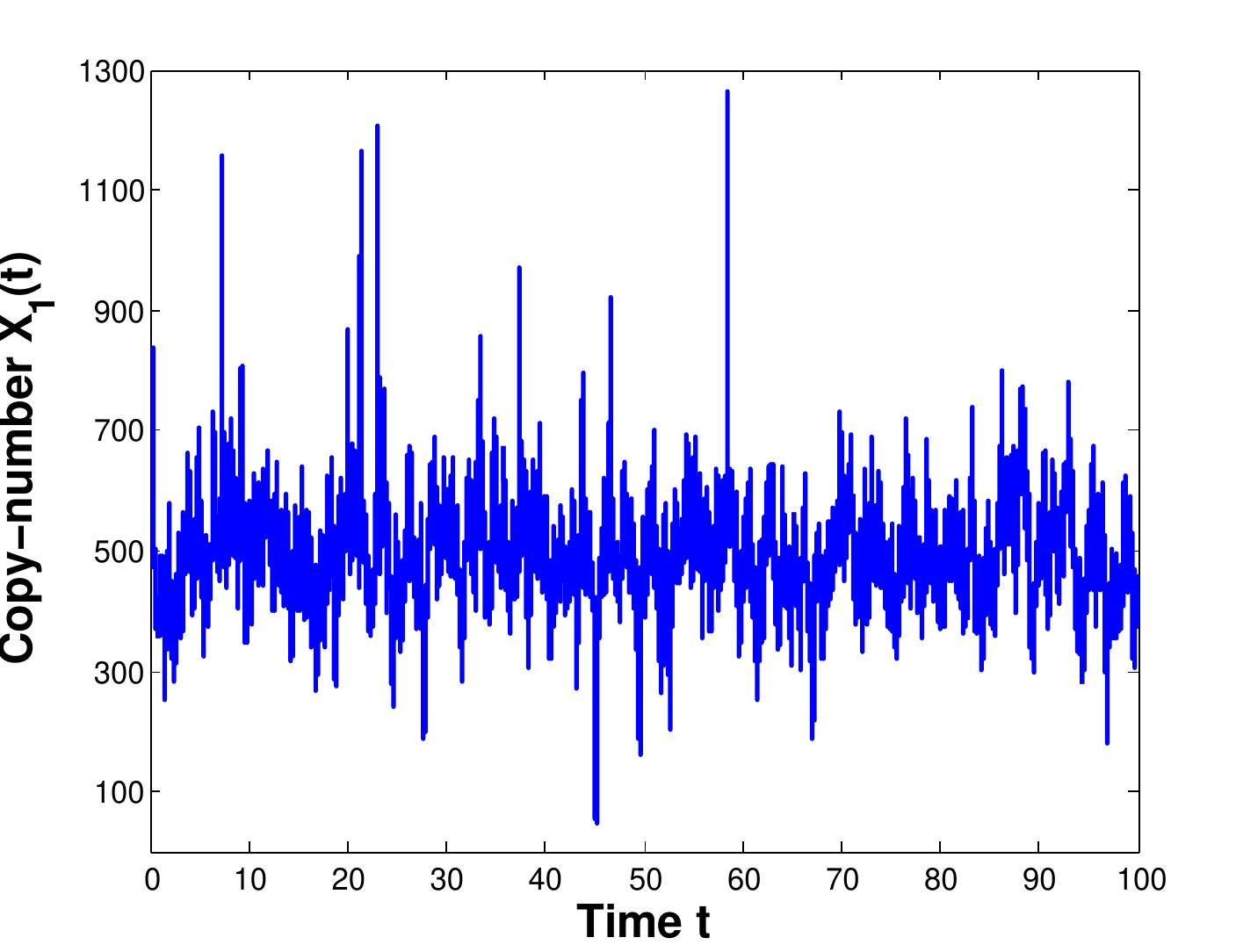}
}
\vskip -5.0cm
\leftline{\hskip 0.4cm (c) {\small $\theta = -0.00147$} 
\hskip 3.8cm (d) {\small $\theta = -0.00145$}}
\vskip 4.8cm
\centerline{
\hskip -2mm
\includegraphics[width=0.46\columnwidth]{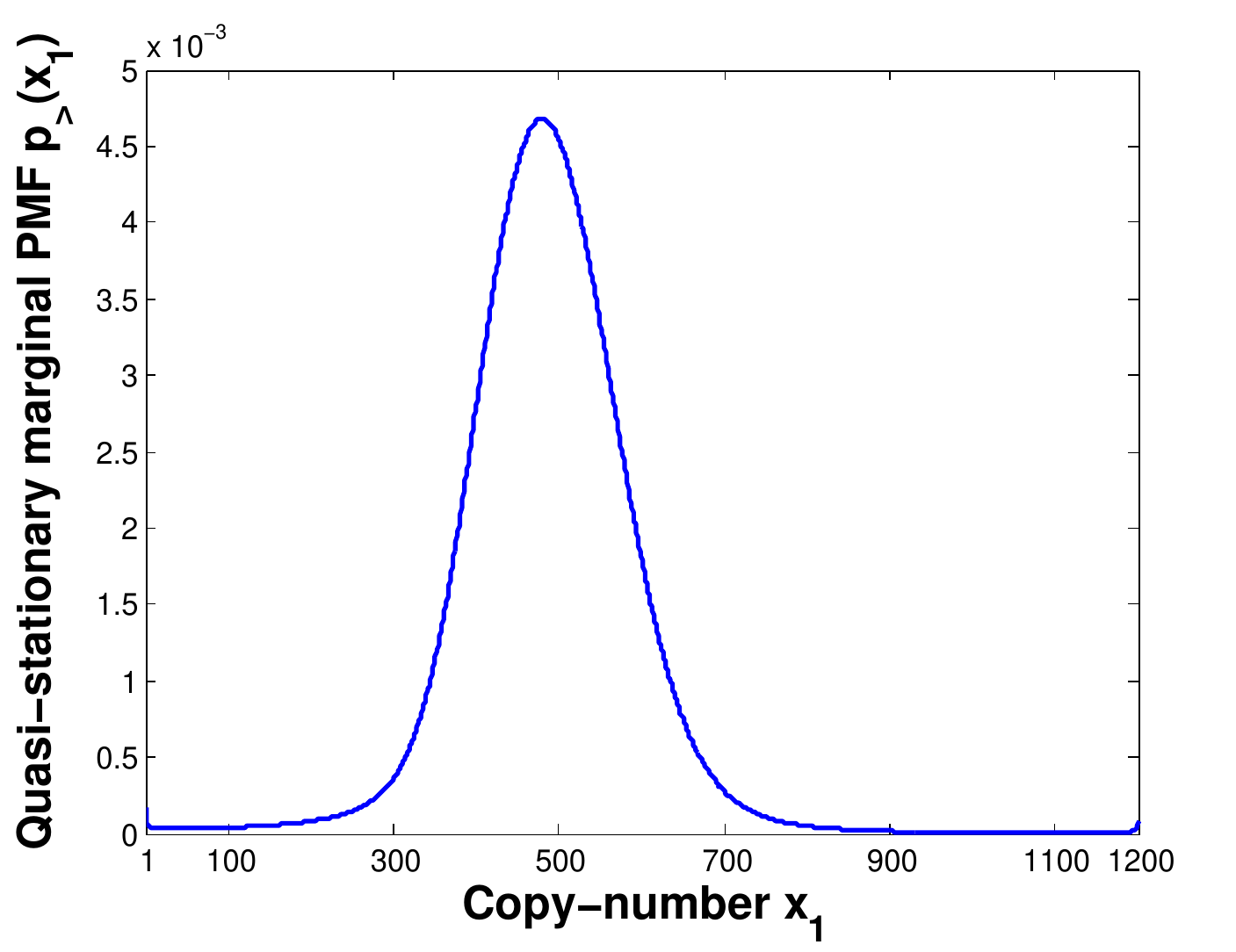}
}
\vskip -4.5cm
\leftline{\hskip 2.7cm (e)}
\vskip 4.1cm
\caption{ 
{\rm (a)--(b)} 
{\it Numerical solutions of the kinetic ODE system 
given by~$(\ref{eq:bicyclicXT})$ before and after the 
bifurcation, where in {\rm (b)} the trajectory initiated 
near the stable limit cycle $L_1$ is shown in purple, 
while the one initiated near $L_3$ in red. \hfill\break
{\rm (c)--(d)} 
Sample paths generated by the Gillespie stochastic
simulation algorithm applied to the induced reaction 
network~$(\ref{eq:bicyclicXT2net})$ before and after the 
bifuration. \hfill\break
{\rm (e)} 
Approximate quasi-stationary $x_1$-marginal PMF,
obtained by numerically solving the stationary CME,
corresponding to network~(\ref{eq:bicyclicXT2net}),  
on the bounded domain $(x_1,x_2) \in [1,1200] \times [1,1200]$,
for the same parameters values as in (d).\hfill\break
The parameters appearing in~$(\ref{eq:bicyclicXT2coefficients})$ 
are fixed to $a = 1$, $b = -1$, $c = 0.5$, $d = 0.08$, 
$x_1^{*} = -3$, $\mathcal{T}_1 = \mathcal{T}_2 = 1000$, 
$\varepsilon = 0.01$, with the reactor volume $V = 0.5$, and 
$\theta$ as indicated in the plots. 
Coefficients~$(\ref{eq:bicyclicXT2coefficients})$ are multiplied 
by a constant factor of $2 \times 10^{-5}$ (time-rescaling).}}
\label{fig:nonconcentricbicyclic}
\end{figure}

\section{Summary} \label{sec:summary}
In the first part of the paper, in Section~\ref{sec:properties}, we have 
presented theoretical results regarding oscillations, 
oscillation-related bifurcations and multistability 
in the planar quadratic kinetic ODEs~(\ref{eq:kinetic}),
which are (appropriately) bounded in the nonnegative quadrant. 
Such ODEs are used in applications to
describe the deterministic dynamics of concentrations of two 
biological/chemical species, with
at most quadratic interactions. While the kinetic ODEs~(\ref{eq:kinetic})
inherit many properties from the more general 
planar quadratic ODEs~(\ref{eq:polynomial}), some properties, which
are of biological/chemical relevance, are not necessarily inherited. 
For example, we have formulated the following open problem:
while general planar quadratic ODEs~(\ref{eq:polynomial}) may display bicyclicity 
(a coexistence of two stable oscillatory attractors), is the same true
for the kinetic planar quadratic ODEs~(\ref{eq:kinetic})?

In Section~\ref{sec:constructions}, building upon the results 
from Section~\ref{sec:properties}, and using the results from~\cite{Me} 
and Appendix~\ref{app:xfactorable},
we have constructed two reaction networks, with the deterministic
dynamics described by planar cubic kinetic ODEs. 
The first network is given by~(\ref{eq:homoclinic1net}),
and, at the deterministic level, displays a homoclinic bifurcation, and a coexistence of 
a stable critical point and a stable limit cycle (mixed bistability). 
The second network is given by~(\ref{eq:bicyclicXT2net}), and, at the deterministic level,
displays a multiple limit cycle 
bifurcation, and a coexistence of two stable limit cycles (bicyclicity).
The phase planes of the kinetic ODEs induced by the first network before and 
after the bifurcation are shown in Figures~\ref{fig:phaseplanes}(a) 
and~\ref{fig:phaseplanes}(b), respectively, while for the second network in 
Figures~\ref{fig:phaseplanes}(c) and \ref{fig:phaseplanes}(d).

In Figure~\ref{fig:homoclinic}, we have compared the deterministic and 
stochastic solutions corresponding to the first reaction network~(\ref{eq:homoclinic1net}),
with the rate coefficients such that the deterministic solutions are close to the
homoclinic bifurcation. Analogously, in Figure~\ref{fig:nonconcentricbicyclic},
we have done the same for reaction network~(\ref{eq:bicyclicXT2net}),
when the deterministic solutions are close to the multiple limit cycle bifurcation.
In both Figures~\ref{fig:homoclinic} and~\ref{fig:nonconcentricbicyclic}, we observe
qualitative differences between the deterministic and stochastic dynamics. In particular,
the stochastic dynamics in Figure~\ref{fig:homoclinic} may display stochastic switching
near the deterministic bifurcation. Furthermore, the dynamics of both networks 
are not affected qualitatively by the 
deterministic bifurcation sharply at the bifurcation point.

In Section~\ref{sec:intro}, we have outlined the statistical inference problem,
consisting of detecting and classifying cycles (oscillations) in noisy
time-series, and we have put forward networks~(\ref{eq:homoclinic1net})
and~(\ref{eq:bicyclicXT2net}) as suitable test problems. Network~(\ref{eq:homoclinic1net})
poses two inference challenges: firstly, let us consider the scenario shown in
Figures~\ref{fig:homoclinic}(d)--(f). In this case, the relative separation
between the two deterministic attractors is larger. 
Consequently, at the stochastic level, the corresponding marginal 
probability mass function (PMF), shown in Figure~\ref{fig:homoclinic} (f), is bimodal.
However, the left peak, corresponding to the deterministic limit cycle, is much smaller
than the right peak, corresponding to the deterministic critical point (a node).
Using the shape of the marginal PMF, as put forward in~\cite{NOI1},
one cannot conclude the presence of a noisy limit cycle. Let us note that, by
sufficiently increasing the distance between the two attractors, the left PMF peak from 
Figure~\ref{fig:homoclinic}(f) approximately vanishes, making the inference
problem even harder. On the other hand, using the covariance function 
(and spectral analysis), as put forward in~\cite{NOI1}, may also be limited, as the noisy
time-series spends a smaller amount of time near the deterministic limit cycle, as demonstrated in 
Figure~\ref{fig:homoclinic}(e). Secondly, let us consider the scenario shown in
Figures~\ref{fig:homoclinic}(a)--(c), when the relative separation
between the two deterministic attractors is smaller. In this case, it may be a challenge to 
infer that there are two distinct attractors `hidden' in the time-series
shown in Figure~\ref{fig:homoclinic}(b), and the PMF shown in 
Figure~\ref{fig:homoclinic}(c). The fact that the PMF in Figure~\ref{fig:homoclinic}(c)
is a non-Gaussian may be used as an indication of a certain dynamical complexity. 
The problem becomes more difficult for network~(\ref{eq:bicyclicXT2net}), 
with two stable deterministic limit cycles `hidden' in the noisy time-series shown in
Figure~\ref{fig:nonconcentricbicyclic}(d), and in the PMF shown in
Figure~\ref{fig:nonconcentricbicyclic}(e). Let us note that the PMF is
approximately Gaussian, and this persists for a wide range of larger 
reactor volumes.

\vskip 4.9mm
\noindent
\textbf{Acknowledgments:} The authors would like to thank the Isaac 
Newton Institute for Mathematical Sciences, Cambridge, for support 
and hospitality during the programme ``Stochastic Dynamical Systems 
in Biology: Numerical Methods and Applications'' where work on this 
paper was undertaken. This work was supported by EPSRC grant no 
EP/K032208/1. This work was partially supported by a grant from 
the Simons Foundation. Tom\'a\v{s} Vejchodsk\'y would like to acknowledge
the institutional support RVO 67985840. Radek Erban would also like to 
thank the Royal Society for a University Research Fellowship.
\vskip 5 mm

\appendixtitleon
\appendixtitletocon
\begin{appendices}
\section{$\!\!\!\!\!\!\!$:$\!$ perturbed $x$-factorable transformation} 
\label{app:xfactorable}
\begin{definition} \label{def:xft}
Consider applying an $x$-factorable transformation, as defined in~\cite{Me}, on~(\ref{eq:polynomial}), and then adding to the resulting 
right-hand side a zero-degree term $\varepsilon \mathbf{v}$, with $\varepsilon \ge 0$ and vector $\mathbf{v} = (1,1)^{\top}$, resulting in
\begin{align}
\frac{\mathrm{d} \mathbf{x}}{\mathrm{d} t} & 
= \varepsilon \mathbf{v} + \mathcal{X}(\mathbf{x}) 
\boldsymbol{\mathcal{P}}(\mathbf{x}; \, \mathbf{k})
\, = \, \varepsilon \mathbf{v} + (\Psi_{\mathcal{X}} \boldsymbol{\mathcal{P}})(\mathbf{x}; \, \mathbf{k})
\equiv
(\Psi_{\mathcal{X}_{\varepsilon}} \boldsymbol{\mathcal{P}})(\mathbf{x}; \, \mathbf{k}).
\label{eqn:xft} 
\end{align}
Then $\Psi_{\mathcal{\mathcal{X}_{\varepsilon}}}: \mathbb{P}_{2}(\mathbb{R}^{2}; 
\, \mathbb{R}^{2}) \to \mathbb{P}_{3}(\mathbb{R}^{2}; \, \mathbb{R}^{2})$, 
mapping $\boldsymbol{\mathcal{P}}(\mathbf{x}; \, \mathbf{k})$ to 
$(\Psi_{\mathcal{X}_{\varepsilon}} \boldsymbol{\mathcal{P}})(\mathbf{x}; \, \mathbf{k})$, 
is called a \emph{perturbed $x$-factorable transformation} if $\varepsilon \ne 0$. If $\varepsilon = 0$, the transformation reduces to an (unperturbed) $x$-factorable transformation, $\Psi_{\mathcal{\mathcal{X}}} \equiv \Psi_{\mathcal{X}_{0}}$, defined in~\cite{Me}.
\end{definition}

\vskip 1.9mm

\begin{lemma} \label{lemma:xfact}
$(\Psi_{\mathcal{X}_{\varepsilon}} \boldsymbol{\mathcal{P}})(\mathbf{x}; \, \mathbf{k})$ 
from \emph{Defnition~{\rm \ref{def:xft}}} is a \emph{kinetic function}, i.e. 
$(\Psi_{\mathcal{X}_{\varepsilon}} \boldsymbol{\mathcal{P}})(\mathbf{x}; 
\, \mathbf{k}) \in \mathbb{P}^{\mathcal{K}}_{3}(\mathbb{R}_{\ge}^{2}; \, 
\mathbb{R}^{2})$.
\end{lemma}

\vskip 1.5mm

\noindent
{\bf Proof.}
$(\Psi_{\mathcal{X}} \boldsymbol{\mathcal{P}})(\mathbf{x}; \, \mathbf{k})$ is a kinetic function~\cite{Me}. Since, from~(\ref{eqn:xft}), $(\Psi_{\mathcal{X}_{\varepsilon}} \boldsymbol{\mathcal{P}})(\mathbf{x}; \, \mathbf{k}) = \varepsilon \mathbf{v} + (\Psi_{\mathcal{X}} \boldsymbol{\mathcal{P}})(\mathbf{x}; \, \mathbf{k})$, with $\varepsilon \ge 0$ and $\mathbf{v} = (1,1)^{\top}$, it follows that $(\Psi_{\mathcal{X}_{\varepsilon}} \boldsymbol{\mathcal{P}})(\mathbf{x}; \, \mathbf{k})$ is kinetic as well.
\hfill \qed 

We now provide a theorem relating location, stability and type of the positive critical points of~(\ref{eq:polynomial}) and~(\ref{eqn:xft}). 
\begin{theorem}\label{theorem:Xfact2D}
\textit{Consider the ODE 
system~$(\ref{eq:polynomial})$ with positive 
critical points $\mathbf{x}^{*} \in \mathbb{R}_{>}^2$. 
Let us assume that $\mathbf{x}^{*} \in \mathbb{R}_{>}^2$ 
is hyperbolic, and is not the degenerate case between 
a node and a focus, i.e. it satisfies the condition
\begin{align}
\left(\textrm{\emph{tr}} \left(\nabla \boldsymbol{\mathcal{P}}(\mathbf{x}^{*}; \, \mathbf{k}) \right) \right)^2 - 4\textrm{\emph{det}} \left(\nabla \boldsymbol{\mathcal{P}}(\mathbf{x}^{*}; \, \mathbf{k}) \right) & \ne 0, \label{eq:degenerate}
\end{align}
as well as conditions \emph{(ii)} and \emph{(iii)} of \emph{Theorem} 
$3.3$ in~{\rm \cite{Me}}. Then positivity, stability 
and type of the critical point $\mathbf{x}^{*} \in \mathbb{R}_{>}^2$ 
are invariant under the perturbed $x$-factorable 
transformations $\Psi_{\mathcal{X}_{\varepsilon}}$, 
for sufficiently small $\varepsilon \ge 0$.
Assume~$(\ref{eq:polynomial})$ does not have boundary 
critical points. Consider the two-dimensional ODE 
system~$(\ref{eqn:xft})$ with $\varepsilon = 0$, 
and with boundary critical points 
denoted $\mathbf{\bar{x}}^{0} \in \mathbb{R}_{\ge}^2$, 
$\mathbf{\bar{x}}^{0} 
= (\bar{x}_{b,1}^0,\bar{x}_{b,2}^0)$, 
$\bar{x}_{b,1}^0 \bar{x}_{b,2}^0 = 0$. Assume that for 
$i \in \{1,2\}$
\begin{align}
\frac{\partial \mathcal{P}_i(\mathbf{\bar{x}}_b^{0}; \, \mathbf{k})}{\partial x_i} & \ne 0, \, \, \, \, \, \, \, \, \, \text{if } \, \, \,  \bar{x}_{b,i}^0 \, \ne \, 0, \label{eq:boundarycondition1}
\end{align}
and that for some $i \in \{1,2\}$ 
\begin{align}
\mathcal{P}_i(\mathbf{\bar{x}}_b^{0}; \, \mathbf{k}) & > 0, \, \, \, \, \, \, \, \, \, \text{if } \, \, \,  \bar{x}_{b,i}^0 \, = \, 0. \label{eq:boundarycondition2}
\end{align}
Then, the critical point 
$\mathbf{\bar{x}}_b^{0} \in \mathbb{R}_{\ge}^2$ of the 
two-dimensional ODE system~$(\ref{eqn:xft})$ with 
$\varepsilon = 0$ becomes the critical point 
$\mathbf{\bar{x}}_b \notin \mathbb{R}_{\ge}^2$ 
of system~$(\ref{eqn:xft})$ for 
sufficiently small $\varepsilon > 0$.
}
\end{theorem}

\begin{proof}
The critical points of~(\ref{eqn:xft}) are solutions of the following regularly perturbed algebraic equation
\begin{align}
\varepsilon \mathbf{v} + \mathcal{X}(\mathbf{\bar{x}}) 
\boldsymbol{\mathcal{P}}(\mathbf{\bar{x}}; \, \mathbf{k}) & = \mathbf{0}. \label{eq:criticalpoints}
\end{align}
Let us assume $\mathbf{\bar{x}}$ can be written as the power series
\begin{align}
\mathbf{\bar{x}} & = \mathbf{\bar{x}}^0  + \varepsilon \mathbf{\bar{x}}^1 + \mathcal{O}(\varepsilon^2), \label{eq:expansion}
\end{align}
where $\mathbf{\bar{x}}^0 \in \mathbb{R}_{\ge}^{2}$ are the critical points of~(\ref{eqn:xft}) with $\varepsilon = 0$. Substituting the power series~(\ref{eq:expansion}) into~(\ref{eq:criticalpoints}), and using the Taylor series theorem on $\boldsymbol{\mathcal{P}}(\mathbf{\bar{x}}; \, \mathbf{k})$, so that $\boldsymbol{\mathcal{P}}(\mathbf{\bar{x}}^0  + \varepsilon \mathbf{\bar{x}}^1 + \mathcal{O}(\varepsilon^2); \, \mathbf{k}) = \boldsymbol{\mathcal{P}}(\mathbf{\bar{x}}^0; \, \mathbf{k}) + \varepsilon \nabla \boldsymbol{\mathcal{P}}(\mathbf{\bar{x}}^0; \, \mathbf{k})\mathbf{\bar{x}}^1 + \mathcal{O}(\varepsilon^2)$, as well as that $\mathcal{X}(\mathbf{\bar{x}}) = \mathcal{X}(\mathbf{\bar{x}}^0) +  \varepsilon \mathcal{X}(\mathbf{\bar{x}}^1) + \mathcal{O}(\varepsilon^2)$, and equating terms of equal powers in $\varepsilon$, the following system of polynomial equations is obtained:
\begin{align}
\mathcal{O} \left(1 \right): & \;  \mathcal{X}(\mathbf{\bar{x}}^0) \boldsymbol{\mathcal{P}}(\mathbf{\bar{x}}^0; \, \mathbf{k}) \, = 0, \nonumber \\
\mathcal{O}(\varepsilon): & \;  \mathcal{X}(\mathbf{\bar{x}}^0) \nabla \boldsymbol{\mathcal{P}}(\mathbf{\bar{x}}^0; \, \mathbf{k}) \mathbf{\bar{x}}^1 +  \mathcal{X}(\mathbf{\bar{x}}^1) \boldsymbol{\mathcal{P}}(\mathbf{\bar{x}}^0; \, \mathbf{k}) \, = - \mathbf{v}. \label{eq:perturbationsystem}
\end{align}

\emph{Order $1$ equation}. The positive critical points $\mathbf{\bar{x}}^0 \in \mathbb{R}_{>}^{2}$ satisfy $\boldsymbol{\mathcal{P}}(\mathbf{\bar{x}}^0; \, \mathbf{k}) = \mathbf{0}$. Since $\boldsymbol{\mathcal{P}}(\mathbf{x}; \, \mathbf{k})$ has no boundary critical points by assumption, critical points $\mathbf{\bar{x}}_b^0 \in \mathbb{R}_{\ge}^{2}$ with $\bar{x}_{b,i}^0 = 0$, $\bar{x}_{b,j}^0 \ne 0$, $\bar{x}_{b,1}^0 \bar{x}_{b,2}^0 = 0$, $i,j \in \{1,2\}$, satisfy $\mathcal{P}_{i}(\mathbf{\bar{x}}_b^0; \, \mathbf{k}) \ne 0$, $\mathcal{P}_{j}(\mathbf{\bar{x}}_b^0; \, \mathbf{k}) = 0$.

\emph{Order $\varepsilon$ equation}. Vector $\mathbf{\bar{x}}^1$, corresponding to a positive $\mathbf{\bar{x}}^0$, satisfies
\begin{align}
\mathcal{X}(\mathbf{\bar{x}}^0) \nabla \boldsymbol{\mathcal{P}}(\mathbf{\bar{x}}^0; \, \mathbf{k}) \mathbf{\bar{x}}^1  & = - \mathbf{v}, \nonumber
\end{align}
which can be solved provided $\mathbf{\bar{x}}^0$ is a hyperbolic critical point. Vector $\mathbf{\bar{x}}_b^1$, corresponding to a nonnegative $\mathbf{\bar{x}}_b^0$, is given by
\begin{align}
\bar{x}_{b,i}^1 = \begin{cases}
 -(\mathcal{P}_{i}(\mathbf{\bar{x}}_b^0; \, \mathbf{k}))^{-1}, & \textrm{if } \bar{x}_{b,i}^0 = 0,\\
(\frac{\partial \mathcal{P}_{i}(\mathbf{\bar{x}}_b^0; \, \mathbf{k})}{\partial x_i})^{-1} \left( (\mathcal{P}_{j}(\mathbf{\bar{x}}_b^0; \, \mathbf{k}))^{-1} \frac{\partial \mathcal{P}_{i}(\mathbf{\bar{x}}_b^0; \, \mathbf{k})}{\partial x_j} - (\bar{x}_{b,i}^0)^{-1}\right), &  \textrm{if } \bar{x}_{b,i}^0 \ne 0, \nonumber
\end{cases}
\end{align}
from which conditions~(\ref{eq:boundarycondition1}) and~(\ref{eq:boundarycondition2}) follow.
\end{proof}

\section{$\!\!\!\!\!\!\!$:$\!$ bicyclic system with large attractors} \label{app:constructions3}
Consider the following deterministic kinetic equations
\begin{align}
\frac{\mathrm{d} x_1}{\mathrm{d} t} 
& = k_1 + x_1 (- k_{2} + k_{3} x_1 + k_{4} x_2 - k_{5} x_1 x_2), 
\nonumber \\
\frac{\mathrm{d} x_2}{\mathrm{d} t} 
& = k_6 + x_2 (k_{7} - k_{8} x_1 + k_{9} x_2 + k_{10}x_1^2 - k_{11} x_2^2), \label{eq:bicyclicXT2}
\end{align}
with the coefficients 
$\mathbf{k}$ given by
\begin{align}
k_1 & = 10^{-3}, \; \; \; 
k_2 = 10, \; \; \; 
k_3 = 1, \; \; \; 
k_4 = 1, \; \; \; 
k_5 = 0.1, \; \; \; 
k_6 = 10^{-3}, \nonumber \\
k_7 & = 3.7, \; \; \; 
k_8 = 1.9, \; \; \; 
k_9 = 1.01, \; \; \; 
k_{10} = 0.1, \; \; \; 
k_{11} = 0.05.
\label{eq:bicyclicexample2}
\end{align}
The canonical reaction network induced by system~(\ref{eq:bicyclicXT2}), involving two species $s_1$ and $s_2$ and eleven 
reactions $r_1, r_2, \ldots, r_{11}$ under mass-action kinetics, is given by
\begin{align}
& r_1: \;  & \varnothing  &
\xrightarrow[]{ k_1 } s_1 ,  
\; \; \; \; \; \; \; \; \; \; \; \; 
\; \; \; \; \; \; \; \; \; \; \; \; \; \; \;  
\; \; \; \; \; \; \; \, r_6:  
& 
\varnothing  &\xrightarrow[]{ k_6 } s_2,  \nonumber \\
& r_2: \;  & s_1  &
\xrightarrow[]{ k_2 } \varnothing ,  \; \; \; \; \; \; \; \; \; \; \; \; 
\; \; \; \; \; \; \; \; \; \; \; \; \; \; \;  \; \; \; \; \; \; \; \, \, r_7:  
& 
s_2  &\xrightarrow[]{ k_7 } 2 s_2,  \nonumber \\
& r_3: \;  & 2 s_1  &\xrightarrow[]{ k_3 } 3 s_1, \; \; \; \; \;  
\; \; \; \; \; \; \; \; \; \; \; \; \;  \; \; \; \; \;  
\; \; \; \; \; \; \; \; \;    
\, r_8: & s_1 + s_2  &\xrightarrow[]{ k_8 } s_1, \nonumber \\
& r_4: \;  & s_1 + s_2  &\xrightarrow[]{ k_4 } 2 s_1 + s_2,  \; \; \; \; \; \; \; 
\; \; \; \; \; \; \; \; \; \; \; \; 
\; \; \; \; \; \; r_{9}:  
&
2 s_2  &\xrightarrow[]{ k_{9} } 3 s_2, \nonumber \\
& r_5: \;  
& 
2 s_1 + s_2  &\xrightarrow[]{ k_5 } s_1 + s_2, \; \; \; \; 
\; \; \; \; \; \; \; \;  \; \; \; \; \;  \; \; \; \; \; \; \; \; \,  
r_{10}: 
& 2 s_1 + s_2 &\xrightarrow[]{ k_{10} } 2 s_1 + 2 s_2, \nonumber \\
& & & \; \; \; \; 
\; \; \; \; \; \; \; \;  \; \; \; \; \;  \; \; \; \; \; \; \;
\; \; \; \; \; \; \; \;  \; \; \; \; \;  \; \; \; \; \; \; \; \;
 r_{11}: & 3 s_2  &\xrightarrow[]{ k_{11} } 2 s_2.  \label{eq:bicyclicXT2net2}
\end{align} 
In Figure~(\ref{fig:fig5})(a), we show the two stable limit cycles obtained by numerically solving~(\ref{eq:bicyclicXT2}) with parameters~(\ref{eq:bicyclicexample2}). In Figure~(\ref{fig:fig5})(b), in addition to the limit cycles, we also show in blue a representative sample path obtained by applying the Gillespie algorithm on~(\ref{eq:bicyclicXT2net2}). Let us note that~(\ref{eq:bicyclicXT2}) was constructed in a similar fashion as system~(\ref{eq:bicyclicXT}) in Section~\ref{sec:constructions2}, using the results from~\cite{Perko2,Tung}.

\begin{figure}[tb]
\centerline{
\hskip 0mm
\includegraphics[width=0.5\columnwidth]{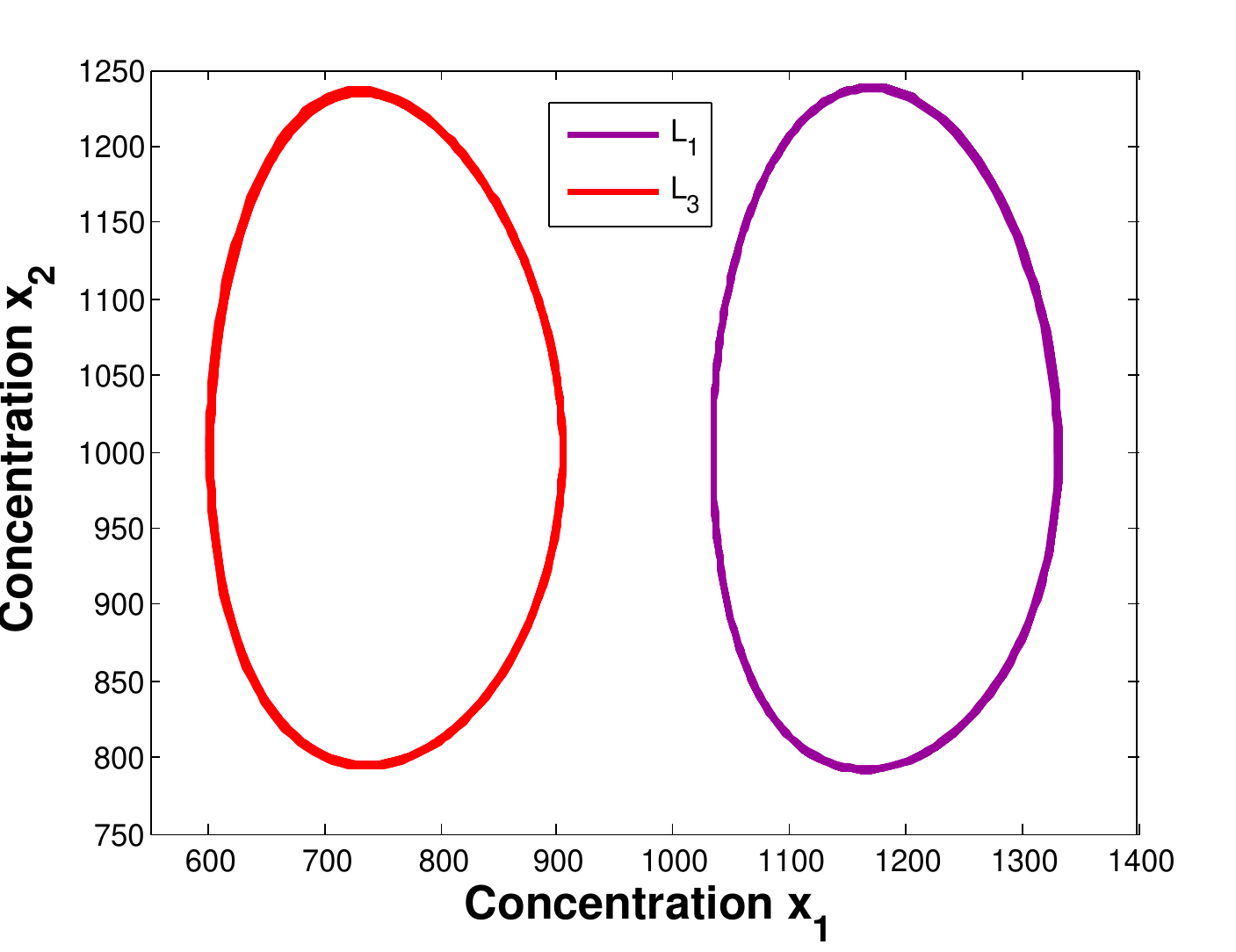}
\hskip 1mm
\includegraphics[width=0.5\columnwidth]{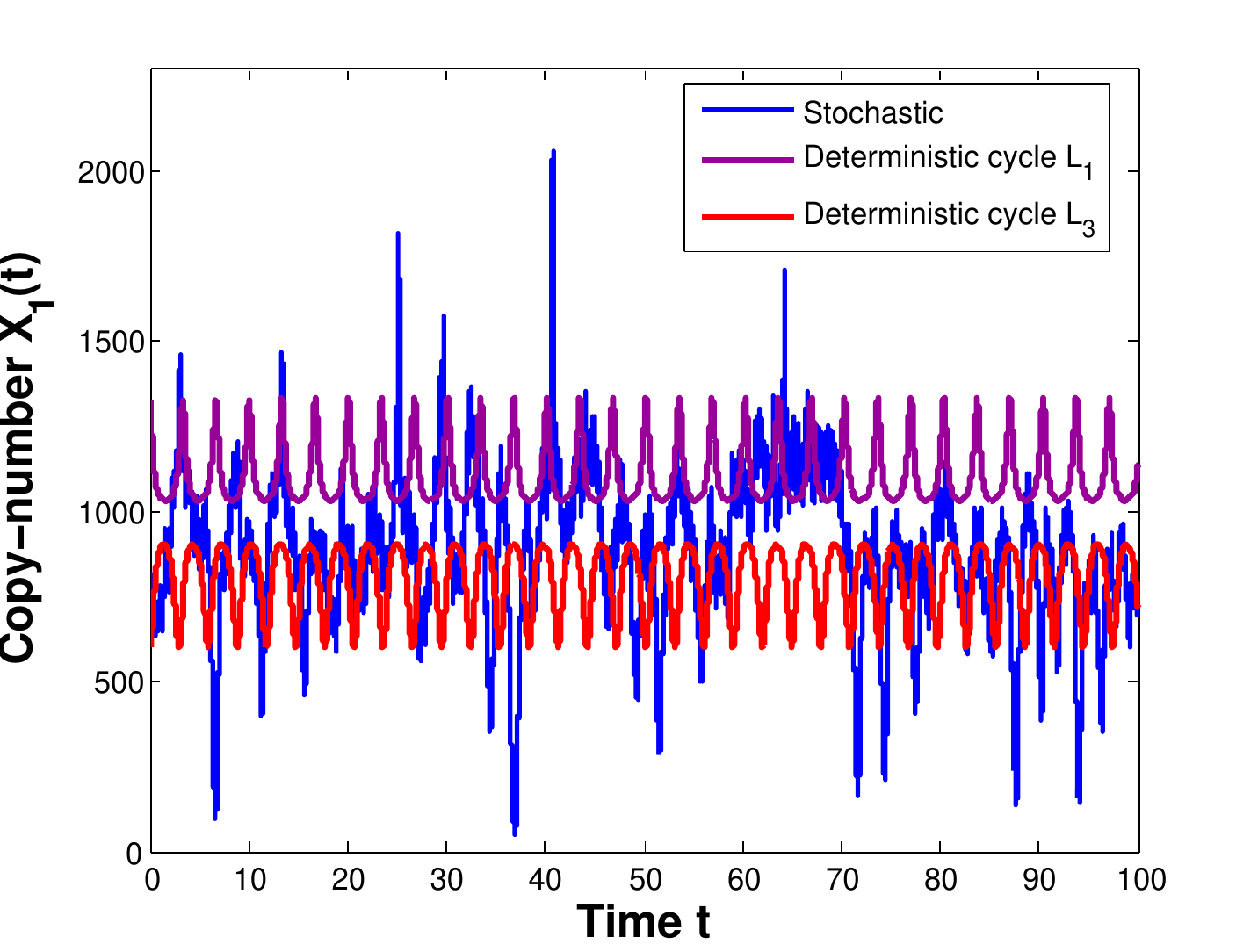}
}
\vskip -5.1cm
\leftline{\hskip -0.3cm (a) \hskip 5.9cm (b)}
\vskip 4.3cm
\caption{ 
{\it Panel {\rm (a)} displays numerically approximated 
stable limit cycles $L_1$ and $L_3$ in the state-space of system~$(\ref{eq:bicyclicXT2})$,
with parameters~$(\ref{eq:bicyclicexample2})$ and reactor volume $V = 100$.
Panel {\rm (b)} displays in blue a representative sample path, 
generated by applying the Gillespie algorithm on the underlying
reaction network~$(\ref{eq:bicyclicXT2net2})$ for the same parameters as in panel {\rm (a)}.
Also shown are two deterministic trajectories, one initiated 
near the limit cycle $L_1$, while the other near $L_3$. One can observe
that the stochastic sample path switches between the two deterministic attractors. 
}
}
\label{fig:fig5}
\end{figure}

\end{appendices}

\end{document}